\documentclass[twoside,leqno,twocolumn]{article}

\usepackage[numbers,sort]{natbib}
\usepackage[letterpaper]{geometry}

\usepackage{ltexpprt}
\usepackage{hyperref}

\usepackage{latexsym,amsmath,textcomp,pifont,marvosym,wasysym,amssymb}

\usepackage{tikz}
\usetikzlibrary{external}

\tikzset{
  external/optimize=true,
  external/figure name=figure_\thechapter_,
}

\usepackage{pgfplots}
\usepgfplotslibrary{external}
\newif\ifpdfplots
\pdfplotstrue

\usepackage[british]{babel}

\usepackage{textcomp}
\usepackage{soul}
\usepackage{amsmath,amsfonts}
\usepackage{pifont}
\usepackage{nicefrac}
\usepackage{textcomp}
\usepackage{mathtools}
\usepackage{array}
\usepackage{tabularx,longtable}
\usepackage{rotating}
\usepackage{verbatim}
\usepackage{etoolbox}
\usepackage{calc}
\usepackage{arydshln}

\usepackage{graphicx}

\usepackage{xcolor}

\usepackage{xspace}

\usepackage{csquotes}
\usepackage{dirtytalk}

\usepackage{float}
\usepackage{pdflscape}

\usepackage{booktabs}
\usepackage{multirow}

\usepackage{datatool}
\DTLsetseparator{ = }
\newcommand{\placeholder}[2]{\DTLfetch{#1}{key}{#2}{value}}

\usepackage{breakcites}

\usepackage{tikz}

\usepackage[
algo2e,
ruled,vlined,
linesnumbered,
noend]{algorithm2e}
\DontPrintSemicolon
\SetAlCapHSkip{0pt}
\SetProgSty{}
\SetArgSty{}

\SetKwComment{tcp}{{\upshape{}/\!/}\,}{}%
\SetCommentSty{textit}

\SetKwProg{Function}{Function}{}{}
\SetKwFunction{LCPCompare}{LCP-Compare}

\SetKwFor{Lock}{lock}{}{}
\SetKwFor{ParallelFor}{for}{do in parallel}{}
\SetKwFor{ParallelWhile}{while}{do in parallel}{}
\SetKwIF{If}{ElseIf}{Else}{if}{then}{else if}{else}{end if}

\newcommand{\arro}[1]{[\, #1 \kern.1em )}
\newcommand{\Oh}[1]{\ensuremath{\mathcal{O}(#1)}\xspace}

\usepackage{wasysym}
\usepackage{marvosym}

\usepackage{xstring}

\usepackage{tikz}

\newif\ifshowcomments
\showcommentstrue

\ifshowcomments
\definecolor{fuchsiapink}{rgb}{1.0, 0.47, 1.0}
\newcommand{\tobi}[1]{{\color{fuchsiapink}[TH: #1]}}
\newcommand{\lars}[1]{{\color{cyan}{[Lars: #1]}}}
\newcommand{\peter}[1]{{\color{blue}[PS: #1]}}
\newcommand{\seb}[1]{{\color{orange}[SS: #1]}}
\newcommand{\todo}[1]{{\textcolor{red}{\bf [TODO]} \emph{#1}}}

\else
\newcommand{\tobi}[1]{}
\newcommand{\lars}[1]{}
\newcommand{\peter}[1]{}
\newcommand{\seb}[1]{}
\newcommand{\todo}[1]{}

\fi

\newcommand{\etal}{{et al}.}

\newcommand{\Increment}{\raisebox{.025ex}{\hbox{\tt ++}}}

\newcommand{\EndOfStatement}{;\quad}

\newcommand{\neighbors}{\ensuremath{{N}}}%
\newcommand{\incnets}{\ensuremath{{I}}}%

\newcommand{\maxdeg}[1]{\ensuremath{\Delta_{#1}}}
\newcommand{\maxsize}[1]{\ensuremath{\Delta_{#1}}}

\newcommand{\meddeg}{\ensuremath{\widetilde{d(v)}}}
\newcommand{\medsize}{\ensuremath{\widetilde{|e|}}}

\newcommand{\Partition}{\ensuremath{{\Pi}}}%
\newcommand{\pinsinpart}{\ensuremath{{\Phi}}}
\newcommand{\adjblocks}{\ensuremath{{R}}}
\newcommand{\avgadjblocks}{\ensuremath{{\overbar{\adjblocks(u)}}}}
\newcommand{\con}{\ensuremath{\lambda}}
\newcommand{\conset}{\ensuremath{\Lambda}}

\newcommand{\nodeblock}[1]{\ensuremath{\Partition[#1]}}

\newcommand{\ocut}{\ensuremath{\mathfrak{f}_c}}%
\newcommand{\ocon}{\ensuremath{\mathfrak{f}_{\lambda-1}}}%
\newcommand{\gain}[2]{\ensuremath{g_{#1}(V_{#2})}}

\newcommand{\cutnets}{\ensuremath{E_{\text{Cut}}(\Partition)}}

\newcommand{\attrgain}{\ensuremath{\Delta}}

\newcommand{\cluster}{\mathcal{C}}

\newcommand{\flownetwork}{\ensuremath{\mathcal{N}}}

\newcommand{\flownodeset}{\ensuremath{\mathcal{V}}}
\newcommand{\flowedgeset}{\ensuremath{\mathcal{E}}}
\DeclareMathAlphabet{\mathpzc}{OT1}{pzc}{m}{n}
\newcommand{\capacity}{\ensuremath{\mathpzc{c}}}
\newcommand{\source}{\ensuremath{s}}
\newcommand{\sink}{\ensuremath{t}}

\newcommand{\overbar}[1]{\mkern 1.5mu\overline{\mkern-1.5mu#1\mkern-1.5mu}\mkern 1.5mu}

\newcommand{\subhypergraph}[1]{\ensuremath{H[#1]}}
\newcommand{\subgraph}[1]{\ensuremath{G[#1]}}

\newcommand{\balancedconstraint}[1]{\ensuremath{L_{#1}}}

\newcommand{\graphdef}[1]{\ensuremath{G^{#1}}}
\newcommand{\routinggraph}{\graphdef{R}}
\newcommand{\globalroutinggraph}{\graphdef{GR}}
\newcommand{\targetarch}{\ensuremath{T}}
\newcommand{\targetarchgraph}{\graphdef{\targetarch}}
\newcommand{\osteiner}{\ensuremath{\mathfrak{f}_{\text{\tiny ST}}}}%
\newcommand{\dist}{\ensuremath{\text{\textsc{dist}}}}%
\newcommand{\deltadist}[1]{\ensuremath{\Delta \dist(e, #1)}}%
\newcommand{\steinertreethres}{\ensuremath{t}}%

\DeclareMathOperator*{\argmin}{arg\,min}

\newcommand{\gmeantime}{geo\-metric mean running time}
\newcommand{\splitatcommas}[1]{%
  \begingroup
  \begingroup\lccode`~=`, \lowercase{\endgroup
    \edef~{\mathchar\the\mathcode`, \penalty0 \noexpand\hspace{0pt plus 1em}}%
  }\mathcode`,="8000 #1%
  \endgroup
}

\newcommand{\ALL}{\textsc{All}}
\newcommand{\SAT}{\textsc{Sat}}

\newcommand{\SPM}{\textsc{Spm}}

\newcommand{\VLSI}{\textsc{Vlsi}}

\newcommand{\Partitioner}[1]{\textsf{#1}} % Favorite so far textsf

\newcommand{\gpp}{\texttt{g++}}

\newif\ifdoubleblindmode
\newcommand{\doubleblind}[2]{%
  \ifdoubleblindmode
    #1
  \else
    #2
  \fi
}

\newcommand{\anonymousauthor}[1]{\doubleblind{Anonymous Author}{#1}}
\newcommand{\anonymousuniversity}[1]{\doubleblind{University,}{#1}}
\newcommand{\anonymouscitycountry}[1]{\doubleblind{City, Country.}{#1}}
\newcommand{\anonymousmail}[1]{\doubleblind{anonymous@author.com}{#1}}
\newcommand{\anonymousurl}[2]{\doubleblind{\url{#2}}{\url{#1}}}

\newtheorem{definition}{Definition}

\begin{document}

\title{\Large A Direct $k$-Way Hypergraph Partitioning Algorithm \\ for Optimizing the Steiner Tree Metric}
\author{\anonymousauthor{Tobias Heuer}\thanks{\anonymousuniversity{Karlsruhe Institute of Technology,} \anonymouscitycountry{Karlsruhe, Germany.} \anonymousmail{tobias.heuer@kit.edu}}}

\date{}

\maketitle

% Copyright Statement
% When submitting your final paper to a SIAM proceedings, it is requested that you include
% the appropriate copyright in the footer of the paper.  The copyright added should be
% consistent with the copyright selected on the copyright form submitted with the paper.
% Please note that "20XX" should be changed to the year of the meeting.

% Default Copyright Statement
% \fancyfoot[R]{\scriptsize{Copyright \textcopyright\ 20XX by SIAM\\
% Unauthorized reproduction of this article is prohibited}}

% Depending on which copyright you agree to when you sign the copyright form, the copyright
% can be changed to one of the following after commenting out the default copyright statement
% above.

%\fancyfoot[R]{\scriptsize{Copyright \textcopyright\ 20XX\\
%Copyright for this paper is retained by authors}}

%\fancyfoot[R]{\scriptsize{Copyright \textcopyright\ 20XX\\
%Copyright retained by principal author's organization}}

% \pagenumbering{arabic}
% \setcounter{page}{1}%Leave this line commented out.

\begin{abstract} % \small\baselineskip=9pt

Minimizing wire-lengths is one of the most important objectives in circuit design.
The process involves initially \emph{placing} the logical units (cells) of a circuit onto a
physical layout, and subsequently \emph{routing} the wires to connect the cells.
Hypergraph partitioning (HGP) has been long used as a placement
strategy in this process. However, it has been replaced by other methods due to the limitation
that common HGP objective funtions only optimize wire-lengths \emph{implicitly}.
In this work, we present a novel HGP formulation
that maps a hypergraph $H$, representing a logical circuit, onto a routing layout represented by a
weighted graph $G$. The objective is to minimize the total length of all wires induced by the hyperedges
of $H$ on $G$. To capture wire-lengths, we compute minimal Steiner trees - a metric commonly used in
routing algorithms.

% The hypergraph partitioning problem (HGP) is to partition the nodes of a hypergraph into $k$ equally-sized blocks such that
% a metric defined on the hyperedges is minimized.
% Hypergraph partitioning (HGP) finds application in problems such as
% mapping logical circuits onto physical chips to minimize wire-lengths (VLSI design) or mapping (hyper)graphs onto
% processors in a distributed system to minimize communication costs.
% The target systems, such as routing layouts in VLSI design or the communication network between processors, can be often represented as graphs.
% Placing highly-connected clusters as closely together as possible on these systems can lead to significant benefits.
% \seb{Ich glaube das sollte noch besser herausgearbeitet werden. A la: Für HGP ist die Struktur der Partition egal...}
% However, existing objective functions for HGP focus on minimizing the number of blocks spanned by each hyperedge, overlooking
% the potential advantage of placing highly-connected clusters close to each other on the target system.
% For example, wires connecting the cells of a logical circuit are longer if the corresponding
% cells are assigned to regions on a physical layout that are far away from each other.
% We therefore propose a novel HGP formulation that considers the target system's structure by mapping a
% hypergraph $H$ onto a weighted graph $G$.
% The objective is to minimize the total weight of all Steiner trees induced by the hyperedges of $H$ on $G$.
% The Steiner tree metric accurately models wire-lengths in VLSI design and communication costs in distributed
% systems.

For this formulation, we present the first direct $k$-way multilevel mapping algorithm
that incorporates techniques used by the highest-quality partitioning algorithms.
We contribute a greedy mapping algorithm to compute an initial solution and three refinement algorithms
to improve the initial mapping: Two move-based local search heuristics (based on label propagation and the FM algorithm)
and a refinement algorithm based on max-flow min-cut computations.

Our experiments demonstrate that our new
algorithm achieves an improvement in the Steiner tree metric by $7\%$ (median) on VLSI instances
when compared to the best performing partitioning algorithm that optimizes the mapping in a
postprocessing step.
Although computing Steiner trees is an NP-hard problem,
we achieve this improvement with only a $2$--$3$ times slowdown in partitioning time compared to optimizing the
connectivity metric.
% Moreover, our algorithm also improves the state-of-the-art for mapping graphs
% onto hierarchical processor architectures.

\end{abstract}

% \begin{center}
% \begin{tabular}{l||l|c|c|c}
% Part                        & Caretaker  & Actual Pages & Target Pages & Done?  \\\hline
% \midrule
% Abstract                    & Tobi       & --           & --           & \no \\
% Introduction                & Tobi       & --           & 0.75         & \yes  \\
% Preliminaries               & Tobi       & 0.66         & 0.5          & \yes \\
% Related Work                & Tobi       & --           & 0.25         & \yes \\
% \midrule
% Steiner Tree Metric         & Tobi       & --           & 1.5          & \yes \\
% Algorithm Overview          & Tobi       & --           & 0.75         & \yes \\
% \midrule
% Initial Partitioning        & Tobi       & --           & 0.5          & \yes \\
% Gain Definitions            & Tobi       & --           & 0.25         & \yes \\
% Label Propagation           & Tobi       & --           & 0.25         & \yes \\
% Attributed Gains            & Tobi       & --           & 0.25         & \yes \\
% FM (Gain Table)             & Tobi       & --           & 0.75         & \yes \\
% Flows                       & Tobi       & --           & 0.5          & \yes \\
% \midrule
% Experiments                 & Tobi       & --           & 3.5          & \yes \\
% Conclusion                  & Tobi       & --           & 0.25         & \no \\
% \midrule
% \textbf{Overall}            &            &              & 10 \\
% \end{tabular}
% \end{center}

\section{Introduction}

The design of complex integrated circuits is a key driver of innovation in modern technology.
The advancements in the field are attributed not only to progress in semiconductor technology but also to
algorithmic innovations enabling the realization of very large-scale circuits. The physical design of a modern chip
involves three stages: In modeling stage, a logical circuit is designed that consists of logical units
(cells) interconnected by wires. In the placement phase, the cells are assigned to
non-overlapping locations within a chip area. Subsequently, during the routing phase, the cells
are interconnected using well-separated wires.
There exist multiple metrics to evaluate the quality of a physical layout,
but a common approach is to minimize the total length of all wires.
This effectively reduces signal delays, power consumption, and the layout area of the chip~\cite{HeldKRV11,ALPERT-SURVEY, GrotschelMW97}.
The placement phase plays an important role in this process as the locations of cells significantly influence
the achievable wire-lengths in the routing phase.
% Therefore, placement algorithms aim to minimize an
% estimate for the total wire-length as, e.g., the half-perimeter of the bounding box around the cells
% connected by a wire~\cite{HeldKRV11,MarkovHK12}.

Algorithms based on hypergraph partitioning have been a key approach for placing cells on a chip in the past.
A hypergraph is a generalization of a graph where a \emph{hyperedge} can connect more than two nodes.
This makes hypergraphs an appropriate model to represent logical circuits since a wire can connect an abitrary number of cells.
The \emph{hypergraph partitioning} problem (HGP) is to partition the node set of a hypergraph into $k$ blocks of
roughly the same size while simultanously minimizing an objective function defined on the hyperedges.
% HGP is NP-hard for various objective functions~\cite{LENGAUER,DBLP:journals/tcs/GareyJS76} and
% it is also unlikely that a constant factor approximations exists~\cite{DBLP:conf/mfcs/Feldmann12}.
HGP is used to partition a logical circuit into densely-connected blocks by minimizing the number of blocks spanned by each hyperedge
(\emph{connectivity} metric). These blocks are subsequently assigned to subregions of the layout area for which the process is repeated recursively
until the regions are small enough~\cite{ShahookarM91,CaldwellKM00,HuangK97}.
An end-case placer then maps the cells to locations on the chip~\cite{CaldwellKM00}.
In recent years, partitioning-based placement techniques have fallen out of fashion since they were outperformed by numerical methods~\cite{XuGA11,pattison2016gplace}.
This can be attributed to the fact that minimizing the connectivity metric
only minimizes wire-lengths \emph{implicitly}~\cite{wang2017survey,HuangK97,MarkovHK12}.

\paragraph{Contributions.}
% \seb{Es kommt nicht ganz gut raus, was es schon gab und was deine Contribution ist. Kommt die Steiner Tree Metric komplett von dir?}
% In this work, we present a HGP model that minimizes wire-lengths \emph{explicitly}.
In this work, we present a HGP formulation that maps the nodes of a hypergraph $H$
onto the nodes of a routing layout represented by a weighted graph $G$ such
that the total length of all wires spanned by the hyperedges of $H$ on $G$ is minimized.
% -- a commonly used metric in routing algorithms~\cite{GrotschelMW97,HeldKRV11,TangLCX20}.
To capture wire-lengths, we compute minimal Steiner trees induced by the hyperedges on $G$ --
a metric commonly used in routing algorithms~\cite{GrotschelMW97,HeldKRV11,TangLCX20}.
% which is known to be an NP-hard problem~\cite{SMT-NP-HARD}.
Unlike similar HGP formulations for wire-length minimization~\cite{RoyLM06,IMF,THETO}, we approach
the problem from a graph-theoretical perspective that does not rely on geometric information.
This makes it applicable to a wide range of similar problems.
For example, the target graph $G$ can also represent the communication links and
costs between processors in a distributed computing cluster.

Our main algorithmic contribution is the \emph{first} direct $k$-way multilevel mapping
algorithm for optimizing the Steiner tree metric.
This includes an efficient \emph{gain table} data structure that stores and maintains
the gain values for all possible node moves. The gain table is used by our implementation of the
FM algorithm to search for improvement. We further devise a novel \emph{flow network model}
to optimize the Steiner tree metric via max-flow min-cut computations.
Additionally, we implement a greedy mapping algorithm to find an initial solution.

% an initial solution, which we subsequently refine using
% two moved-based local search heuristics, namely label propagation and FM refinement,
% along with a refinement technique based on maximum flows.

\paragraph{Results.}
We integrated the mapping algorithm into the state-of-the-art shared-memory partitioner
\Partitioner{Mt-KaHyPar}~\cite{MT-KAHYPAR-JOURNAL} and
evaluated it on $175$ graphs and hypergraphs with up to $100$ million edges/pins.
The experiments demonstrates that our algorithm outperforms traditional partitioning algorithms that optimize the mapping in a postprocessing step.
More precisely, we achieve an improvement in the Steiner tree metric by $7\%$ in the median on \VLSI~instances
when compared to the best performing traditional approach.
Optimizing the Steiner tree metric comes at the cost of a $2$--$3$ times slowdown in partitioning time
compared to optimizing the connectivity metric. However, it is almost an order of magnitude faster than
the connectivity optimization code of \Partitioner{KaHyPar}~\cite{KAHYPAR-JOURNAL} and \Partitioner{hMetis}~\cite{HMETIS}
when using ten threads.
Moreover, our new algorithm computes better mappings than comparable implementations for mapping graphs
onto hierarchical processor architectures.

\paragraph{Outline.}
After describing basic notations and reviewing related work in Section~\ref{sec:preliminaries}
and~\ref{sec:related_work}, we introduce the Steiner tree metric in Section~\ref{sec:steiner_tree}.
Section~\ref{sec:overview}--\ref{sec:flows} presents our direct $k$-way mapping algorithm for optimizing
the Steiner tree metric, which we compare to other algorithms in Section~\ref{sec:experiments}.
Section~\ref{sec:conclusion} concludes the work.

\begin{figure*}[!ht]
	\centering
	\includegraphics[width=0.9\textwidth]{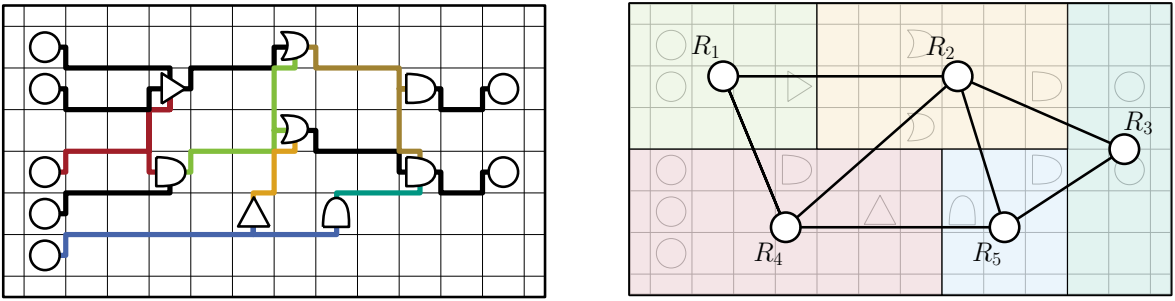}
	\caption{A routing of wires on a physical layout (left) and a global routing graph with five regions (right).  }\label{fig:routing_example}
  \vspace{-0.31cm}
\end{figure*}

\section{Preliminaries}\label{sec:preliminaries}

\paragraph{Hypergraphs.}
A \emph{weighted hypergraph} $H=(V,E,c,\omega)$ is defined as a set of  $n$ nodes $V$ and a set of $m$
hyperedges $E$ (also called \emph{nets})
with node weights $c:V \rightarrow \mathbb{R}_{>0}$ and net weights $\omega:E \rightarrow \mathbb{R}_{>0}$, where each net $e$ is
a subset of the node set $V$. The nodes of a net are called its \emph{pins}.
We extend $c$ and $\omega$ to sets in a natural way, i.e., $c(U) :=\sum_{u \in U} c(u)$ and $\omega(F) :=\sum_{e \in F} \omega(e)$.
A node $u$ is \emph{incident} to a net $e$ if $u \in e$.
$\incnets(u) := \{e \mid u \in e\}$ is the set of all incident nets of $u$.
The set $\neighbors(u) := \{ v \mid \exists e \in E : \{u,v\} \subseteq e\}$ denotes the neighbors of $u$.
Two nodes $u$ and $v$ are \emph{adjacent} if $v \in \neighbors(u)$.
The \emph{degree} of a node $u$ is $d(u) := |\mathrm{I}(u)|$.
The \emph{size} $|e|$ of a net $e$ is the number of its pins.
% Nets of size one are called \emph{single-pin} nets.
% We call two nets $e_i$ and $e_j$ \emph{identical} if $e_i = e_j$.
We denote the number of pins of a hypergraph with $p := \sum_{e \in E} |e| = \sum_{v \in V} d(v)$.
Given a subset $V' \subset V$, the \emph{subhypergraph}
$\subhypergraph{V'}$ is defined as $\subhypergraph{V'}:=(V', \{e \cap V' \mid e \in E : e \cap V' \neq \emptyset \}, c, \omega')$
where $\omega'(e \cap V')$ is the weight of hyperedge $e$ in $H$.
% The \emph{bipartite graph representation} $G_x := (V \cup E, E_x)$~\cite{BIPARTITE-GRAPH,HYPERGRAPH-KL}
% of an unweighted hypergraph $H = (V,E)$ contains the nodes and nets of $H$ as node set and for each pin $u \in e$, we add an undirected edge $\{u,e\}$ to $E_x$.
% More formally, $E_x := \{\{u,e\} \mid \exists e \in E: u \in e\}$.

An undirected and weighted graph $G = (V,E,c,\omega)$ can be considered as a hypergraph where each net contains only two pins (also called an \emph{edge}).
We therefore use the same notation for undirected graphs.
%Therefore, the definitions and notations for hypergraphs also apply to undirected graphs.
We define the weight of an edge $e = \{u,v\} \in E$ as $\omega(u,v) := \omega(e)$.
% If $\{u,v\} \notin E$, then $\omega(u,v) = 0$.

\paragraph{Partitions.}
% A \emph{clustering} $\cluster = \{C_1, \ldots, C_l\}$ of a hypergraph $H = (V,E,c,\omega)$ is a partition of the node set $V$ into disjoint subsets.
% A cluster $C_i$ is called a \emph{singleton} cluster if $|C_i| = 1$.
% A node contained in a singleton cluster is called \emph{unclustered}.
A \emph{$k$-way partition} of a hypergraph $H$ is a partition of its node set $V$ into $k$ disjoint blocks
$\Partition = \{V_1, \ldots, V_k\}$.
A $2$-way partition is also called a \emph{bipartition}.
We denote the block to which a node $u$ is assigned by $\nodeblock{u}$.
% We note that this notation is not correct in a mathematical sense. However, we did not introduce a seperate
% function for the block assigments of nodes since $\nodeblock{u}$ leads often to a better understanding
% of the presented data structures and algorithms.
For each net $e$, $\pinsinpart(e,V_i) := |e \cap V_i|$ denotes the number of pins of net $e$ in block $V_i$.
The \emph{connectivity set} $\conset(e) := \{V_i \mid  \pinsinpart(e, V_i) > 0\}$ contains the blocks to which pins of net $e$ are assigned.
Similarly, $\adjblocks(u) := \{ \nodeblock{v} \mid v \in \neighbors(u) \}$ are the blocks to which node $u$ is adjacent.
The \emph{connectivity} $\con(e)$ of a net $e$ is $\con(e) := |\conset(e)|$.
A net is called a \emph{cut net} if $\con(e) > 1$.
Similarly, a node $u$ is a \emph{boundary node} if $|\adjblocks(u)| > 1$.

\paragraph{Balanced Hypergraph Partitioning.}

The \emph{balanced hypergraph partitioning problem} (HGP) is to
find a $k$-way partition $\Partition$ of a hypergraph $H$ that minimizes an objective function defined on the hyperedges, where
each block $V' \in \Partition$ satisfies the \emph{balance constraint}: $c(V') \leq \balancedconstraint{\max} := (1+\varepsilon) \lceil \frac{c(V)}{k} \rceil$
% \footnote{The
% $\lceil\cdot\rceil$ in this definition ensures that there is always a feasible solution for inputs with unit node weights. However, this does not hold for general weighted inputs as finding
% a balanced solution for its own is an NP-hard problem~\cite{garey1979computers}. There exists several alternative definitions~\cite{KAMINPAR,DEEP-BALANCE-PAPER}, but no commonly
% accepted way how to deal with feasibility. In this work, we use the original definition since our benchmark instances are unweighted.}
for some \emph{imbalance ratio} $\varepsilon \in (0,1)$.
If $\Partition$ satisfies the balance constraint, we call $\Partition$ \emph{$\varepsilon$-balanced}.
% \seb{den rest vom Satz weglassen?} or just say \emph{balanced} or \emph{feasible} when $\varepsilon$ is clear from the context.
For $k = 2$, we refer to the problem as the \emph{bipartitioning} problem.
The two most prominent objective functions are the \emph{cut-net} metric $\ocut(\Partition) := \sum_{e \in \cutnets} \omega(e)$
and \emph{connectivity} metric $\ocon(\Partition) := \sum_{e \in \cutnets} (\lambda(e) - 1) \cdot \omega(e)$
where $\cutnets$ denotes the set of all cut nets.
% The cut-net metric directly generalizes the edge cut metric from graphs to
% hypergraphs and minimizes the weight of all cut hyperedges.
% The connectivity metric additionally considers the number of blocks connected by a net and thus more accurately models
% the communication volume for parallel computations~\cite{PATOH} (e.g., for the parallel sparse matrix-vector multiplication).
HGP is NP-hard for both objective functions~\cite{LENGAUER,DBLP:journals/tcs/GareyJS76}.
% \seb{weglassen: and it is also unlikely that a constant factor approximation exists~\cite{DBLP:conf/mfcs/Feldmann12}.}
% The connectivity value of a net is bounded by its size and the number of blocks.
% Hence, the connectivity metric reverts to the cut-net metric for graph partitioning
% ($|e| = 2$ for all $e \in E$) and bipartitioning ($\con(e) \le 2$ for all $e \in E$).

\paragraph{Steiner Trees.}

For a weighted graph $G = (V,E,\omega)$ and $N$ \emph{terminal sets} $T_1, \ldots, T_N \subseteq V$, the
\emph{Steiner tree packing problem} asks for $N$ edge-disjoint trees $S_1, \ldots, S_N \subseteq E$
such that each $S_i$ spans the nodes in $T_i$ and $\sum_{i = 1}^{N} \omega(S_i)$ is minimal.
For a single terminal set $T$, the problem reverts to the \emph{Steiner tree problem}, which is known
to be NP-hard~\cite{SMT-NP-HARD}. If $T = V$, the problem can be reduced to the \emph{minimum spanning tree problem} (MST),
which is solvable in polynomial time. There exists a $2$-approximation for the Steiner tree problem by computing an MST
on the metric completion $G^T := (T, T \times T, d)$ of $G$ induced by $T$, where $d(u,v)$ is the
shortest distance from $u$ to $v$ in $G$~\cite{KouMB81}.

\paragraph{VLSI Design.}
A logical circuit can be represented as a hypergraph $H = (V,E)$, where the \emph{netlist} $E = \{e_1, \ldots, e_m\}$ describes
the wiring between the cells.
In VLSI design, cells are assigned to locations on a 2D grid, and subsequently
interconnected through disjoint \emph{routing channels}, as illustrated in Figure~\ref{fig:routing_example} (left).
The routing area can be modeled with a \emph{routing graph} $\routinggraph = (V,E,\omega)$ where nodes are
cell locations and each edge connects two cells $u$ and $v$ with distance $\omega(u,v)$.
The routing problem can be formulated as a Steiner tree packing problem on $\routinggraph$ with terminal sets
$e_1, \ldots, e_m$~\cite{GrotschelMW97,HeldKRV11}.

The process of finding a feasible routing is typically decomposed into two steps: \emph{global} and \emph{detailed} routing~\cite{TangLCX20}.
Global routing partitions the layout area into subregions and solves the routing problem on a coarser approximation of $\routinggraph$.
This approximation is represented by a \emph{global routing graph} $\globalroutinggraph = (V,E,c,\omega)$, where
the node set are the regions and two regions are connected if they share a common border, as illustrated in Figure~\ref{fig:routing_example} (right).
Additionally, each edge $\{u,v\} \in E$ has a capacity $c(u,v)$ denoting the maximum number of wires that can cross the
border~\cite{TangLCX20}, and is associated with the distance $\omega(u,v)$ between the centers of both regions~\cite{GrotschelMW97}.
The global routing problem can be also expressed as a slightly modified version of the Steiner tree packing problem on $\globalroutinggraph$
with all nets that span multiple regions as terminal sets and the relaxation that each edge $\{u,v\} \in E$ can be used
in up to $c(u,v)$ different trees~\cite{GrotschelMW97}.
In detailed routing, the cells within each region are routed.

\section{Related Work}\label{sec:related_work}

\paragraph{Hypergraph Partitioning.}
There is a vast amount of literature on hypergraph partitioning for which we refer the reader to
extensive survey articles~\cite{ALPERT-SURVEY,GRAPH-SURVEY,PAPA-MARKOV,KAHYPAR-DIS,HYPERGRAPH-SURVEY}.
In addition to VLSI design, HGP has applications in minimizing the communication volume in parallel scientific
simulations~\cite{PATOH,DBLP:conf/ipps/CatalyurekA01,DBLP:conf/sc/CatalyurekA01}, storage sharding in distributed
databases~\cite{schism,sword,SHP,clay}, simulations of distributed quantum circuits~\cite{gray2021hyper,andres2019automated},
and as a branching strategy in satisfiability solvers~\cite{SATApplication}.
The most notable software packages are \Partitioner{hMetis}~\cite{HMETIS}, \Partitioner{PaToH}~\cite{PATOH},
\Partitioner{KaHyPar}~\cite{KAHYPAR-JOURNAL} (sequential), \Partitioner{Mt-KaHyPar}~\cite{MT-KAHYPAR-JOURNAL} (shared-memory),
and \Partitioner{Zoltan}~\cite{ZOLTAN} (distributed-memory).

% In the following section, we will introduce the Steiner tree metric for HGP that reverts to the objective
% function of the \emph{general process mapping problem} for plain graphs.
% We therefore will review additional literature in the relevant sections.

\paragraph{Wire-Length Minimization via Partitioning.}
Previous attempts to minimize wire-lengths via HGP include methods that optimize metrics
such as the half-perimeter of the bounding boxes around each net~\cite{THETO,IMF} (HPWL), or Rectilinear Steiner Minimal trees induced
by the nets on the physical layout~\cite{RoyLM06} (RSMT, Steiner trees based on Manhatten distance).
These algorithms use recursive bipartitioning (RB) to recursively assign the cells to subregions of the chip.
% hypergraph and assigning the resulting blocks to subregions of the chip. The bipartitioning process
% is then recursively repeated on the subregions until the regions are small enough for detailed routing.
However, before each bipartitioning step, the weights of the nets are adjusted to ensure that the cut-net metric is equal
to either the HPWL or RSMT metric.
A min-cut algorithm is then used to obtain a bipartition.
However, these approaches can not be generalized to abitrary values of $k$ and it has been shown that
RB can produce $k$-way partitions far from optimal~\cite{SimonTeng97}.
Moreover, these methods rely on geometric information, making them less
applicable to other problems.
% Moreover, the algorithm for optimizing RSMT~\cite{RoyLM06} uses only
% approximations for the optimal Steiner trees.

Huang~\etal~\cite{HuangK97} implemented a direct $k$-way partitioning algorithm that optimizes a metric based on
computing MSTs. However, their algorithm is restricted to at most four blocks since the space complexity grows exponentially with $k$.

\paragraph{Process Mapping.}
For a terminal set $T = \{u,v\}$ with two nodes, finding the optimal Steiner tree reverts to computing the shortest distance between them.
A metric based on shortest distances has gained attention in the context of mapping
graphs onto communication networks, which models the communication links and
costs between processors in a distributed computing system.
The communication costs induced by an edge $\{u,v\}$ is the shortest
distance between the two processors to which $u$ and $v$ are assigned.
This problem is known as the \emph{general process mapping problem} (GPMP).
Algorithms for solving GPMP employ either a two-phase approach~\cite{HoeflerS11,JeannotMT14,ChenCHRK06,Kirchbach0T20} or directly optimize
the mapping during partitioning~\cite{PellegriniR96,WalshawC01,KAFFPA-IMAP,PARHIP-IMAP}.
The former first computes a $k$-way partition $\Partition$ of the
input graph optimizing the number of cut edges. Afterwards, each block of $\Partition$
is assigned to a node of the communication network, to which we refer
as the \emph{one-to-one process mapping problem} (OPMP, graph and communication network have the same number of nodes). Algorithms for solving OPMP include
greedy construction~\cite{GlantzMN15,MuellerMerbach} and local search methods based on exchanging the block
assignment of two nodes~\cite{Brandfass2013,Heider72,Kirchbach0T20}.

\section{The Steiner Tree Metric}\label{sec:steiner_tree}

% \seb{gehört der anfang dieser section nicht zu related work?}
% \seb{Hier wird nicht klar was deine contribution ist und was prior work. Ich nehme an die Steiner Tree metric kommt von dir? Das sollte besser rausgearbeitet werden.}

In this work, we revisit wire-length minimization via hypergraph partitioning,
but approach it from a graph-theoretical perspective that does not require any geometric
information. The following definition can be seen as a generalization of the general process mapping problem
from graphs to hypergraphs, making it applicable to a wide range of applications.
\begin{definition}[The Steiner Tree Metric]
Given a weighted hypergraph $H = (V,E,c,\omega)$ and a target graph $\targetarchgraph = (V_\targetarch, E_\targetarch, \omega_\targetarch)$,
the task is to find an $\varepsilon$-balanced mapping $\Partition: V \rightarrow V_\targetarch$ that minimizes the
Steiner tree metric
\begin{equation*}
  \osteiner := \sum_{e \in E} \dist(\conset(e)) \cdot \omega(e)
\end{equation*}
where $\dist(\conset(e))$ is the weight of the minimal Steiner tree connecting the blocks $\conset(e)$ spanned by
net $e \in E$ in $\targetarchgraph$.
\end{definition}
The target graph $\targetarchgraph$ can represent a global routing graph using
an abitrary distance metric, or a communication network modeling communication links and costs between processors.
Moreover, the Steiner tree metric reverts to the objective function for GPMP on plain graphs and to the
connectivity metric for a complete target graph with unit edge weights.

Note that the mapping $\Partition$ represents a $k$-way partition into $k := |V_\targetarch|$ blocks.
We therefore apply partition-related notations and definitions also to mappings.
% The Steiner tree partitioning formulation models the global routing problem, while implicitly considering the capacity constraint,
% as a good solution would also minimize the number of wires that cross the border of two regions.
% However, it provides more flexibility since cells are not fixed to locations.

\paragraph{Calculating $\osteiner$.}
The computation of the Steiner tree metric involves solving an NP-hard problem for each net $e \in E$.
However, it has been observed that real-world VLSI instances contain a large number of small nets, with
only a few nets connecting several hundred or thousands of cells~\cite{CaldwellKKM99,PAPA-MARKOV,KAHYPAR-DIS}.
This observation suggests that computing a small number of optimal Steiner trees suffices to obtain a near-optimal
approximation of $\osteiner$.
We therefore precompute all Steiner trees for terminal sets up to a certain size $\steinertreethres$
using an algorithm proposed by Dreyfus and Wagner~\cite{DreyfusW71}. The algorithm is based
on dynamic programing and has a running time of $\Oh{k^3 + k^2(2^{t} - t) + k(3^{t} - 2^{t+1} + 3)}$.
% This enables us to compute $\dist(\conset(e))$ for all nets $e \in E$ with $\con(e) \le \steinertreethres$.
For nets $e \in E$ with $\con(e) > \steinertreethres$, we calculate a $2$-approximation of the minimal Steiner tree in
$\Oh{\con(e)^2 + \con(e)\log{\con(e)}}$ time by using Prim's MST algorithm~\cite{KouMB81}.
We additionally cache the result of a MST computation in a hash table for subsequent retrievals.

% We compute an unique index for the $\dist(\conset(e))$ values based on $\conset(e)$ to
% access precomputed or cached entries.
We store the connectivity set $\conset(e)$ of
each net $e \in E$ as a bitset of size $k$, and use \emph{count-leading-zeroes} instructions to iterate
over $\conset(e)$ and \emph{pop-count} instructions to compute $\con(e)$.
Therefore, retrieving a precomputed or cached $\dist(\conset(e))$ value takes
$\Oh{\con(e)}$ time since we have to calculate a hash value based on $\conset(e)$.
In the following, we assume that $\dist(\conset(e))$ can be calculated in constant time
to simplify complexity discussions.

% It is worth noting that the running time and space complexity of the precomputation
% step limits the size of the target graph $\targetarchgraph$.
% \seb{komplett weglassen und vllt stattdessen in die conclusion?: Moreover, it is important to note that our focus is not the implementation of a placement algorithm, as this would
% exceed the scope of a single paper. We acknowledge both problems as a potential avenue for future research and
% focus on optimizing $\osteiner$ for moderately-sized target ($|V_\targetarch| \le 64$).}
% \seb{Wie realistisch/anwendingsrelevant ist $|V_\targetarch| \le 64$?}

\paragraph{Gain Concept.}
The \emph{gain} of a node move is defined as the change in the objective function when we move a node $u$
from its current block $\Partition[u]$ to a target block $V_t$. This concept is central for all move-based local search
algorithms~\cite{FM,KL,HYPERGRAPH-KL}.
Clearly, the Steiner tree metric $\osteiner$ only changes if a block is added or removed from the connectivity set $\conset(e)$
of a net $e$. We therefore introduce $\deltadist{\Lambda'} := \dist(\conset(e)) -  \dist(\Lambda')$ which is the difference
in the weight of the minimal Steiner tree if $\conset(e)$ changes to $\Lambda' \subseteq \Partition$.
If $\deltadist{\Lambda'} > 0$, then $\Lambda'$ induces a shorter Steiner tree than $\conset(e)$ on $\targetarchgraph$.

A node move from block $V_s$ to $V_t$ can either remove $V_s$, add $V_t$, or replace $V_s$ with $V_t$ in $\conset(e)$.
We can compute the state of $\conset(e)$ after the node move based on the pin count values $\pinsinpart(e,V_s)$ and
$\pinsinpart(e,V_t)$ of net $e$. For example, if $\pinsinpart(e,V_s) = 1$ and $\pinsinpart(e,V_t) > 0$ then the node move
removes $V_s$ from $\conset(e)$.
We can define the gain $\gain{u}{t}$ of moving a node $u$ from its current block $V_s := \Partition[u]$ to $V_t$
for the Steiner tree metric as follows:

\begin{equation*}
  \begin{aligned}
    \gain{u}{t} := & \sum_{\substack{e \in \incnets(u) \\ \pinsinpart(e,V_s) = 1 \\ \pinsinpart(e,V_t) > 0}} \deltadist{\conset(e) \setminus \{V_s\}} \cdot \omega(e) + \\
                   & \sum_{\substack{e \in \incnets(u) \\ \pinsinpart(e,V_s) > 1 \\ \pinsinpart(e,V_t) = 0}} \deltadist{\conset(e) \cup \{V_t\}} \cdot \omega(e)      +  \\
                   & \sum_{\substack{e \in \incnets(u) \\ \pinsinpart(e,V_s) = 1 \\ \pinsinpart(e,V_t) = 0}} \deltadist{\conset(e) \setminus \{V_s\} \cup \{V_t\}} \cdot \omega(e)
    \end{aligned}
    \vspace{-0.175cm}
\end{equation*}

Our algorithm stores and maintains the pin count values of each net and block.
Thus, we can compute $\gain{u}{t}$ in $\Oh{|\incnets(u)|}$ time and the target block with the highest gain for a node $u$ in
$\Oh{|\adjblocks(u)||\incnets(u)|}$ time.

\section{Algorithm Overview}\label{sec:overview}
% \seb{An einigen Stellen hast du hier details über die Parallelisierung drin, die nicht wirklich für das Paper relevant sind. IMHO verwirrt das eher. Ich würde daher versuchen klarer rauszustellen, dass du das vorhandene Mt-KaHyPar framework benutzt, neue IP und refinement algorithmen einbaust und den Rest so lässt wie er ist.}
Algorithm~\ref{pseudocode:multilevel} shows the high-level structure of our mapping algorithm for optimizing the Steiner tree metric.
It implements the multilevel scheme which is the most successful method to solve the partitioning
problem~\cite{KAHYPAR-JOURNAL,KAFFPA,PATOH,HMETIS,MT-KAHYPAR-JOURNAL,KAMINPAR}.
The technique consists of three phases.
First, the hypergraph is \emph{coarsened} to obtain a hierarchy of successively smaller and structurally similar
approximations of the input hypergraph by \emph{contracting} clusters of highly-connected nodes (Line~\ref{multilevel:coarsening_pass}--\ref{multilevel:contraction}).
Once the hypergraph is small enough, an \emph{initial mapping} of the smallest hypergraph onto $\targetarchgraph$ is computed (Line~\ref{multilevel:initial_mapping}).
Subsequently, the contractions are reverted level-by-level, and, on each level, \emph{local search} heuristics are
used to improve the mapping from the previous level (Line~\ref{multilevel:uncoarsening}--\ref{multilevel:flows}).
% \todo{maybe shortly explain benefit of multilevel}

We implemented our algorithm in the shared-memory hypergraph partitioner \Partitioner{Mt-KaHyPar}~\cite{MT-KAHYPAR-D,MT-KAHYPAR-Q,MT-KAHYPAR-FLOWS,MT-KAHYPAR-JOURNAL}.
% \Partitioner{Mt-KaHyPar} outperforms most of the existing partitioning algorithms already in its fastest configuration for optimizing the connectivity metric $\ocon$,
% while its highest-quality configuration is the state-of-the-art for HGP with regards to solution quality~\cite{MT-KAHYPAR-JOURNAL}.
\Partitioner{Mt-KaHyPar} already implements the multilevel scheme shown in Algorithm~\ref{pseudocode:multilevel} and
provides an interface for implementing new objective functions without touching the internal implementation
of the refinement algorithms.
Thus, we reuse many of its components.
The main algorithmic contributions include an \emph{initial mapping} algorithm that maps
the smallest hypergraph onto $\targetarchgraph$, as well as a \emph{gain table} that stores and maintains the gain values
for all possible node moves.
We use the gain table to lookup gain values for node moves in constant time.
Additionally, we present the first \emph{flow network model} for optimizing $\osteiner$
via max-flow min-cut computations.
We structure the following algorithm description according to the different phases of the multilevel scheme.
Our focus will be on highlighting our algorithmic contributions, while providing only a high-level description of the reused components.
For more details on \Partitioner{Mt-KaHyPar}, we refer the reader to Ref.~\cite{MT-KAHYPAR-JOURNAL}.

\begin{algorithm2e}[!t]
  % \footnotesize
\KwIn{$H = (V,E)$ and $\targetarchgraph = (V_\targetarch, E_\targetarch, \omega_\targetarch)$}
\KwOut{Mapping $\Partition$ of $H$ onto $\targetarchgraph$}
%\SetEndCharOfAlgoLine{}
\caption{The Mapping Algorithm
%\seb{Instead of "$V_s$ has too many nodes": $|V_s| >$ threshold?}
%\tobi{We do not introduce any tuning parameters in this pseudocode.}
%\seb{If we treat $\mathcal{H}$ as a stack, we can get rid of $n$ completely}
%\tobi{I have no opinion on that. We would then need an extra statement to pop an element from the stack.}
}\label{pseudocode:multilevel}

$H_1 := (V_1, E_1) \gets H\EndOfStatement n \gets 1$\;
\While (\label{multilevel:coarsening_pass}) {$V_n$ has too many nodes}{
  $\cluster \gets \FuncSty{ComputeClustering}(H_n)$ \;
  $H_{n + 1} \gets H_n.\FuncSty{Contract}(\cluster) \EndOfStatement \Increment n$\; \label{multilevel:contraction}
}

$\Partition \gets$ \FuncSty{InitialMapping}($H_n, \targetarchgraph$)\; \label{multilevel:initial_mapping}
\For (\label{multilevel:uncoarsening}) {$i = n - 1$ {\bf down to} $1$}{
  $\Partition \gets$ project $\Partition$ onto $H_i$\;
  $\FuncSty{LabelPropagationRefinement}(H_i,\Partition)$\; \label{multilevel:lp} % \tcp*[r]{finds easy improvements by moving single nodes}
  $\FuncSty{FMRefinement}(H_i,\Partition)$\;   % \tcp*[r]{finds short and non-trivial move sets}
  $\FuncSty{FlowBasedRefinement}(H_i,\Partition)$\; \label{multilevel:flows} % \tcp*[r]{global optimization finding long and complex move sets}
}

% \lIf { $\Partition$ is not balanced } { $\FuncSty{Rebalancing}(H,\Partition)$ }
\Return{$\Partition$}
\end{algorithm2e}

\section{Coarsening}\label{sec:coarsening}

Our coarsening algorithm repeatedly finds a clustering $\cluster$ of the nodes and subsequently
contracts it until the hypergraph is small enough ($\approx 160k$ nodes).

The clustering algorithm initially assigns each node to a \emph{singleton} cluster ($\cluster = V$).
Then, it iterates over all nodes $u \in V$ in parallel, and assigns each node $u$ to the cluster $C \in \cluster$
that maximizes the heavy-edge rating function $r(u,C) := \sum_{e \in \incnets(u) \cap \incnets(C)} \frac{\omega(e)}{|e|-1}$.
This function is commonly used in the partitioning literature~\cite{PATOH,HMETIS,KAHYPAR-K} and prefers clusters connected to $u$ via
a large number of heavy nets with small size.

The contraction step replaces each cluster $C_i \in \cluster$ with one supernode $u_i$ with weight $c(u_i) = \sum_{v \in C_i} c(v)$.
For each net $e \in E$, it replaces each pin $v \in e$ with its corresponding supernode.
% After the replacement, multiple occurrences
% of the same supernode in a net are discarded.
% Additionally, we remove single-pin nets ($|e| = 1$) and nets that became identical to each other (e.g., $e_i = e_j$)
% except for one representative at which we aggregate their weights.
% This step can significantly reduce the number of nets in the
% coarser approximation and speedups other algorithmic components.

\section{Initial Mapping}
\label{sec:initial_mapping}

Our initial mapping algorithm follows the two-phase approach. We first compute a $k$-way partition
of the hypergraph using \Partitioner{Mt-KaHyPar}'s initial partitioning algorithm for optimizing
the connectivity metric $\ocon$. Subsequently, we contract each block
of the partition to obtain a one-to-one process mapping problem (OPMP).
To solve OPMP, we implement a greedy strategy to construct an initial mapping, which we
subsequently refine using a local search algorithm. Recall that the hypergraph and target graph
have the same number of nodes when solving OPMP ($k = n$).

% \paragraph{Initial Partitioning.}
% The initial partitioning algorithm of \Partitioner{Mt-KaHyPar} recursively bipartitions a hypergraph
% until it is partitioned into the desired number of blocks. The recursive bipartitioning calls are performed in parallel,
% using work-stealing to account for load imbalances. To obtain a bipartition, it uses a portfolio of nine different
% bipartitioning techniques from which the best partition is used as initial solution.

\paragraph{Greedy Mapping.}
We generalized the greedy construction method of Glantz~\etal~\cite{GlantzMN15} for OPMP on graphs to hypergraphs.
The algorithm starts by assigning a seed node $u$ of $H$ to a block $b$ of $\targetarchgraph$
with the smallest communication volume ($b := \argmin_{u \in V_\targetarch} \sum_{\{u,v\} \in E_\targetarch} \omega(u,v)$).
Afterwards, each neighbor $v \in N(u)$ is inserted into a priority queue (PQ) with the weight of all nets connecting $v$ to the actual
partial mapping as rating. In each step, the algorithm repeatedly (i) extracts a node $u \in V$ from the PQ,
(ii) assigns it to a \emph{free} block with the highest gain, and (iii) updates the ratings of all neighbors $v \in N(u)$ in the PQ.
A block is free if no node has been previously assigned to it.
The complexity of the algorithm is dominated by the second step, which is
$\Oh{k\sum_{u \in V} |\adjblocks(u)||\incnets(u)|} = \Oh{k^2p}$.
We compute several initial mappings in parallel using each node $u \in V$ as seed node and take the best out of all runs as solution.

\paragraph{Local Search.}
We implement the Kerninghan-Lin algorithm~\cite{KL,HYPERGRAPH-KL} to improve an initial mapping.
The algorithm proceeds in passes. At the beginning of each pass, all node pairs are inserted into a PQ
with the gain of exchanging their block IDs as rating. This can be done in $\Oh{k^2(p + \log{k})}$ time.
The algorithm then repeatedly extracts a node pair $(u,v)$ from the PQ and applies the exchange operation
on the mapping. Before applying the operation, we recompute the exchange gain
and reinsert $(u,v)$ into the PQ when the recomputed gain differs from the gain in the PQ. This approach lazily updates
the gain values in the PQ. We ignore $(u,v)$, if either $u$ or $v$ has already been moved before.
The algorithm also performs exchange operations that intermediately worsen $\osteiner$ and
is therefore able to escape from local optima. At the end, we revert to the best seen solution during the pass.
We perform several passes until no further improvement is possible.

If the recomputed gain always matches the gain in the PQ, the running time of a pass is $\Oh{p + k^2\log{k}}$.
The algorithm extracts $k^2$ node pairs from the PQ, but applies only $\nicefrac{k}{2}$ on the mapping for which we
recompute the gain in $\Oh{|\incnets(u)| + |\incnets(v)|}$ time.
In the worst case, the gains of all node pairs are updated in each step, resulting in a running
time of $\Oh{k^3(p + \log{k})}$.

\section{Label Propagation Refinement}\label{sec:lp}

The most widely used refinement technique in parallel partitioning
algorithms is label propagation~\cite{JOSTLE,PARMETIS,PARHIP,MT-KAHIP,KAMINPAR,PARKWAY-2,BIPART}.
The algorithm works in rounds. In each round, it iterates over all nodes in parallel, and whenever it
visits a node $u$, it moves it to the block $V_t$ maximizing its move gain $\gain{u}{t}$.
The algorithm only performs moves with postive gain and therefore cannot escape from local optima.

A fundamental challenge for parallel refinement algorithms is that the gain of a node move can change between its
initial calculation and actual execution due to concurrent node moves in its neighborhood~\cite{PARMETIS}. We therefore
double-check the gain of a node move by comparing it to an \emph{attributed gain} value that we compute
based on syncronized data structure updates~\cite{MT-KAHYPAR-D,MT-KAHYPAR-JOURNAL}. If we move a node $u$ from
$V_s$ to $V_t$, we iterate over all nets $e \in \incnets(u)$ and update $\pinsinpart(e,V_s)$, $\pinsinpart(e,V_t)$,
and $\conset(e)$. This is done while holding a spin lock for net $e$. The updated values are used to compute
the attributed gain $\attrgain$. If $\attrgain < 0$, we immediately revert the node move as we assume
that it has worsen the solution quality.

\section{FM Local Search}\label{sec:fm}

The Fiduccia-Mattheyses (FM) algorithm~\cite{FM} is the most widely used local search technique in sequential partitioning
algorithms~\cite{KAHYPAR-JOURNAL,KAFFPA,PATOH,HMETIS}. The algorithm works similar to the Kerninghan-Lin algorithm but
performs one node move at a time instead of exchanging node pairs. Our parallel version of the algorithm is based on
the \emph{localized multi-try FM} algorithm of Sanders and Schulz~\cite{KAFFPA,MT-KAHIP}.

\paragraph{Algorithm Overview.}
The parallel FM implementation proceeds in several rounds (at most $10$).
In each round, all boundary nodes are inserted into a globally shared task queue $Q$. The threads then poll a fixed number
of nodes (= 25) from $Q$ that they use as seed nodes for their \emph{localized FM} searches. The localized FM searches initializes
a PQ storing the move with the highest gain for each seed node. The algorithm then repeatedly extracts a node move from the PQ,
applies it to the mapping, and inserts or updates the gain of all neighbors of the moved node in the PQ.
% The searches of different threads are non-overlapping, i.e.,
% threads acquire exclusive ownership of nodes, while hyperedges can touch multiple searches.
The localized FM algorithm proceeds until the PQ is empty
or it becomes unlikely to find further improvements~\cite{kaspar}.
At the end, the algorithm reverts to the best seen solution.
We repeatedly start localized FM searches until the task queue $Q$ is empty.
In each round, each node is moved at most once.

The localized FM searches use a globally shared \emph{gain table} to lookup gain values in constant time.
The gain table stores and maintains the gain values for all possible node moves ($nk$ entries).
We use the gain table to insert new node moves or update gain values in the PQ.
It has been shown that this technique speeds up the FM algorithm significantly in practice~\cite{KAHYPAR-DIS,GOTT-DIS}.
In the following, we present an efficient algorithm to update gain table entries for the Steiner tree metric.

\paragraph{Delta-Gain Updates.}

If we move a node $u$ from block $V_s$ to $V_t$, we iterate over all nets $e \in \incnets(u)$ and decide
based on $\pinsinpart(e, V_s)$ and $\pinsinpart(e,V_t)$ if the gain value of a pin $v \in e$ has changed.
% An efficient \emph{delta-gain update} procedure \emph{identifies} these entries and \emph{updates} them in constant time.
% We call this procedure for each net $e \in \incnets(u)$ when we perform the node move on our data structure.
The input for the \emph{delta-gain update} procedure is a copy of the pin count values and connectivity set obtained from
the synchronized data structure updates described in the previous section.

The contribution of net $e$ to the gain value of a pin $v \in e$ can only change if moving $u$
removes $V_s$ or adds $V_t$ to $\conset(e)$ (i.e., $\pinsinpart(e,V_s) = 0$ or $\pinsinpart(e,V_t) = 1$),
or when moving $v$ will now remove $V_s$ (i.e., $\pinsinpart(e,V_s) = 1$ and $\Partition[v] = V_s$), or does not remove $V_t$
from $\conset(e)$ anymore (i.e., $\pinsinpart(e,V_t) = 2$ and $\Partition[v] = V_t$).
Thus, there are four cases that trigger a gain update.
In the former two cases, the gain values of all pins $v \in e$ are affected. In the latter two cases, only
the gain value of the (previously) last pin $v \in e$ in block ($V_t$) $V_s$ changes.

\begin{figure}[t!]
	\centering
	\includegraphics[width=0.49\textwidth]{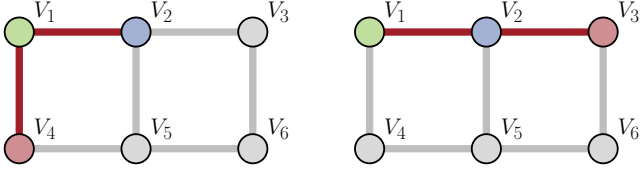}
  \vspace{-0.5cm}
	\caption{After moving the red node from $V_4$ to $V_3$, moving the blue node from $V_2$ to $V_1$ does not
           reduce the weight of the minimal Steiner tree (red edges) anymore.}\label{fig:delta_gain_update}
  \vspace{-0.5cm}
\end{figure}

For the connectivity metric $\ocon$, only the gain value for moving $v$ to $V_s$ or $V_t$ can change if $v \in e$ is affected by a gain update~\cite{KAHYPAR-JOURNAL}.
Unfortunately, this is not the case for $\osteiner$ which we illustrate with an example in Figure~\ref{fig:delta_gain_update}.
Thus, we have to update all $k$ gain table entries of $v$.

If we identify an entry $\gain{v}{j}$ affected by a gain update, we subtract the previous (before moving $u$) and add the actual contribution
(after moving $u$) of net $e$ to the gain value $\gain{v}{j}$.
We compute the previous contribution by reconstructing the pin count values and connectivity set of net $e$ before the node move.
This can be done in constant time by incrementing $\pinsinpart(e,V_s)$ and decrementing $\pinsinpart(e,V_t)$, and adding $V_s$ or removing
$V_t$ from $\conset(e)$ accordingly.

\begin{lemma}
The running time of gain updates for moving all nodes once is $\Oh{k^2p}$.
\end{lemma}

\begin{proof}
The four different update cases are exactly the same as in the original FM algorithm~\cite{FM}.
Fiduccia and Mattheyses showed that they are only triggered a constant number of times per net and block.
The most time consuming update takes $\Oh{k|e|}$ time ($\pinsinpart(e,V_s) = 0$ or $\pinsinpart(e,V_t) = 1$).
Thus, the running time of the gain updates for moving all nodes once is $\Oh{k\sum_{e \in E} k|e|} = \Oh{k^2p}$
\end{proof}

The gain updates for the Steiner tree metric are significantly more expensive compared to the gain updates for the connectivity
metric, which can be performed in $\Oh{kp}$ time~\cite{MT-KAHYPAR-JOURNAL}. As an optimization, we only store and maintain gain values
to the adjacent blocks of a node. This reduces the running time to $\Oh{k\avgadjblocks p}$ where
$\avgadjblocks$ is the average number of adjacent blocks per node. On real-world hypergraphs, we observed running time
much closer to $\Oh{\avgadjblocks p}$ since many nets have small size or connectivity.

\section{Flow-Based Refinement}\label{sec:flows}
Flow-based refinement is considered as the most powerful improvement heuristic for (hyper)graph partitioning and is used
in the highest-quality partitioning algorithms~\cite{KAHYPAR-JOURNAL,KAFFPA,MT-KAHYPAR-FLOWS}.
The technique is based on the well-known max-flow min-cut theorem~\cite{MINCUT-THEOREM} that relates the minimum cut separating two nodes $\source$
and $\sink$ of a graph $G$ to the maximum flow between $\source$ and $\sink$ in the corresponding flow network of $G$.

The flow-based refinement algorithm of \Partitioner{Mt-KaHyPar} works on bipartitions and is
scheduled on different block pairs in parallel to improve a $k$-way partition $\Partition$. To improve a bipartition,
the algorithm extracts a region $B := B_1~\cup~B_2$ with $B_1 \subseteq V_1$ and $B_2 \subseteq V_2$ around the cut nets
of a bipartition $\{V_1,V_2\} \subseteq \Partition$. This yields a flow network $\flownetwork_B$ induced by the region $B$.
The algorithm then uses incremental maximum flow computations on $\flownetwork_B$ to compute a balanced minimum cut that improves $\Partition$.

Max-flow min-cut computations can be used to improve the cut-net metric $\ocut$ on $\flownetwork_B$.
However, we can improve the Steiner tree metric $\osteiner$ on the input hypergraph $H$
with flow computations on $\flownetwork_B$ if the constructed flow network $\flownetwork_B$ satisfies the following equation:

\begin{equation}
\osteiner(\Partition') = \osteiner(\Partition) + \Delta \ocut(\flownetwork_B, \Partition \rightarrow \Partition')
\label{eq:flows}
\end{equation}

Here, $\Delta \ocut(\flownetwork_B, \Partition \rightarrow \Partition')$ is the difference in the cut-net metric
if the actual partition $\Partition$ changes to $\Partition'$ on $\flownetwork_B$.
Note that only nodes contained in $B$ are allowed to change their blocks.
In the following, we show how to construct $\flownetwork_B$ for plain graphs and then generalize the construction
to hypergraphs.

\paragraph{Graph Model.}
Our flow network $\flownetwork_B = (\flownodeset := \{\source,\sink\} \cup B, \flowedgeset, \capacity)$ is a graph with
a dedicated source $\source \in \flownodeset$ and sink $\sink \in \flownodeset$ in which each edge $e \in \flowedgeset$
has capacity $\capacity(e) \ge 0$. Recall that $\dist(\conset(e))$ reverts to shortest distance between the blocks
connected by edge $e$ in $\targetarchgraph$ for plain graphs. The following construction algorithm is illustrated in Figure~\ref{fig:flow_network}.

For an input graph $G$, we obtain $\flownetwork_B$ from the subgraph $\subgraph{V_1\cup V_2}$ induced by the bipartition $\{V_1,V_2\} \subseteq \Partition$
by contracting all nodes in $V_1 \setminus B$ to the source $\source$ and all nodes in $V_2 \setminus B$ to the sink $\sink$.
We set the capacity of each edge $e \in \flowedgeset$ to $\capacity(e) := \dist(\{V_1,V_2\})\cdot \omega(e)$ (blue edges in Figure~\ref{fig:flow_network}).
If a non-cut edge of $\Partition$ becomes a cut edge in $\Partition'$, then the cut-net metric increases by $\dist(\{V_1,V_2\})\cdot \omega(e)$
on $\flownetwork_B$ which is also the increase in the Steiner tree metric on the input graph $G$ (vice versa if we remove an edge from the cut).
However, this does not suffice to satisfy Equation~\ref{eq:flows}.

We have to add additional edges to $\flownetwork_B$ for each edge $e = \{u,v\}$ of $G$ that w.l.o.g.~connect $u \in B$ to a block $\nodeblock{v} \notin \{V_1,V_2\}$
as follows:
\begin{enumerate}
  \item If $u \in V_1$ and moving $u$ from $V_1$ to $V_2$ worsens resp.~improves $\osteiner$ on $G$, we add an edge $e' = \{s,u\}$ resp.~$e' = \{u,t\}$ to $\flowedgeset$
        with capacity $\capacity(e') = |\deltadist{\conset(e) \setminus \{V_1\} \cup \{V_2\}}|\cdot \omega(e)$ (green edges in Figure~\ref{fig:flow_network}).
  \item If $u \in V_2$ and moving $u$ from $V_2$ to $V_1$ worsens resp.~improves $\osteiner$ on $G$, we add an edge $e' = \{u,t\}$ resp.~$e' = \{s,u\}$ to $\flowedgeset$
        with capacity $\capacity(e') = |\deltadist{\conset(e) \setminus \{V_2\} \cup \{V_1\}}|\cdot \omega(e)$ (red edges in Figure~\ref{fig:flow_network}).
\end{enumerate}

The capacities are chosen to accurately reflect the difference in $\osteiner$ on $G$ if $u$ is moved to the other block.
If, e.g., moving $u \in V_1$ to $V_2$ improves $\osteiner$ on $G$, the node move would remove the edge $(u,t)$ in $\flownetwork_B$
from the cut. The difference in $\ocut$ on $\flownetwork_B$ is then the same as the difference in $\osteiner$ on $G$
(analogously for all other cases).

\begin{figure}[t!]
	\centering
	\includegraphics[width=0.49\textwidth]{img/flow_network_construction.pdf}
  \vspace{-0.5cm}
	\caption{The flow network $\flownetwork_B$ induced by a region $B := B_1 \cup B_2$ constructed for a target
           graph $\targetarchgraph$ with six nodes.}\label{fig:flow_network}
  \vspace{-0.3cm}
\end{figure}

\paragraph{Hypergraph Model.}
For an input hypergraph $H$ and a bipartition $\{V_1, V_2\} \subseteq \Partition$, we construct $\flownetwork_B$ similarly to our graph model.
However, the main difference is how we model nets $e$ of $H$ that connect $B$ to a block different from $V_1$ and $V_2$ ($\conset(e) \not\subset \{V_1,V_2\}$).
For such nets, we distingush between two cases: $|e \cap B| = 1$ and $|e \cap B| > 1$.
In the case of $|e \cap B| = 1$, where only one pin $u$ of $e$ is contained in $\flownetwork_B$,
we can use the same approach as in our graph model. If $|e \cap B| > 1$, then
the connectivity set $\conset(e)$ of $e$ can be in one of the following three states if the actual
partition $\Partition$ changes to $\Partition'$: $\conset(e) \setminus \{V_1\} \cup \{V_2\}$,
$\conset(e) \setminus \{V_2\} \cup \{V_1\}$, or $\conset(e) \cup \{V_1, V_2\}$.
This induces three different Steiner trees on $\targetarchgraph$. Thus, we cannot model
the difference in the Steiner tree metric with a single capacity on the nets in $\flownetwork_B$.
We therefore set the capacity of such nets to a lower bound for the actual improvement,
which is
\begin{align*}
  \min( & \deltadist{\conset(e) \setminus \{V_1\} \cup \{V_2\} }, \\
        & \deltadist{\conset(e) \setminus \{V_2\} \cup \{V_1\}})\cdot \omega(e)
\end{align*}
if $\{V_1,V_2\} \subset \conset(e)$, and $\deltadist{\conset(e) \cup \{V_1, V_2\}} \cdot \omega(e)$ if $\{V_1, V_2\} \not\subset \conset(e)$.
Consequently, $\ocut(\flownetwork_B, \Partition \rightarrow \Partition')$ is a lower bound for the actual improvement in the Steiner tree metric
in our hypergraph flow network model. Note that our maximum flow algorithm only works on graph flow networks.
We therefore convert the hypergraph flow network into an equivalent graph model using the \emph{Lawler expansion}~\cite{Lawler}.

\section{Experiments}\label{sec:experiments}

We integrated the multilevel algorithm for optimizing the Steiner tree metric in the shared-memory hypergraph partitioner
\Partitioner{Mt-KaHyPar}\footnote{The code is publicly available from \anonymousurl{https://github.com/kahypar/mt-kahypar}{https://tinyurl.com/fu6sfzp8}}~\cite{MT-KAHYPAR-JOURNAL},
which is implemented in \texttt{C++17}, parallelized using the TBB library~\cite{TBB}, and compiled using \gpp{10.2} with the
flags \texttt{-O3 -mtune=native -march=native}.

\paragraph{Algorithm Configuration.}
We provide two configurations of our algorithm: \Partitioner{Mt-KaHyPar-D} (\textbf{D}efault setting)
and \Partitioner{Mt-KaHyPar-Q} (\textbf{Q}uality setting). \Partitioner{Mt-KaHyPar-Q} extends \Partitioner{Mt-KaHyPar-D}
with flow-based refinement. Moreover, \Partitioner{Mt-KaHyPar-Q} runs all refinement algorithms on each level multiple times
in combination and stops if the relative improvement is less than $0.25\%$. In both configurations, we precompute all Steiner
trees for terminal sets with up to $t = 4$ nodes. We use the concurrent hash table \Partitioner{growt}~\cite{GROWT} to cache
MST computations for nets $e$ with $\con(e) > t$.
% \seb{Die 2 Sätze würde ich hier komplett weglassen und das hoechstens im Future Work aufgreifen.
% Scalability oder irgendwas parallel zu machen ist in dem Paper ja garnicht dein Ziel.}
% Note that the Steiner tree precomputation can become a parallelization
% bottleneck since we run this step sequentially. We therefore do not perform scalability experiments in this evaluation.

\paragraph{Competitors.}
There exists no publicly available algorithm for optimizing $\osteiner$.
We therefore compare \Partitioner{Mt-KaHyPar-D/-Q} to two-phase approaches (2P).
We first compute a $k$-way partition using a hypergraph
partitioning algorithm optimizing the connectivity metric $\ocon$. Subsequently, we contract each
block of the partition to obtain an OPMP, which we solve with our initial mapping algorithm presented
in Section~\ref{sec:initial_mapping}.

We compare \Partitioner{Mt-KaHyPar-D/-Q} to the direct $k$-way version \Partitioner{$k$KaHyPar$_{\text{2P}}$} of
\Partitioner{KaHyPar}~\cite{KAHYPAR-JOURNAL} (highest-quality sequential algorithm), the recursive bipartitioning
version \Partitioner{hMetis-R$_{\text{2P}}$} of \Partitioner{hMetis $2.0$}~\cite{HMETIS} (most widely-used algorithm), as well as
to the connectivity optimization code \Partitioner{Mt-KaHyPar-D$_{\text{2P}}$} (fastest partitioning algorithm) and
\Partitioner{Mt-KaHyPar-Q$_{\text{2P}}$} (comparable quality to \Partitioner{KaHyPar}) of \Partitioner{Mt-KaHyPar}~\cite{MT-KAHYPAR-JOURNAL}.
This algorithm selection is based on a recent study on the time-quality trade-off of partitioning algorithms~\cite{MT-KAHYPAR-JOURNAL}.

\begin{figure*}[!ht]
  \centering
  \ifpdfplots
    \includegraphics{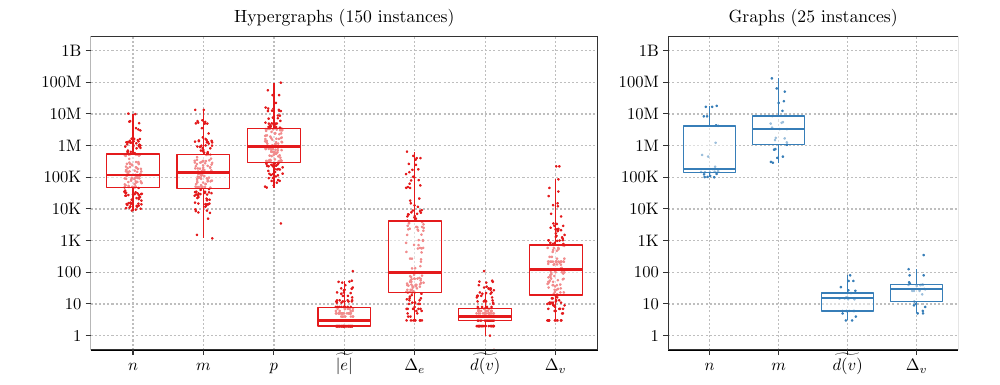}
  \else
    \tikzsetnextfilename{pdf_plots/instances}%
    \input{tikz_plots/instances}%
  \fi
  \vspace{-0.25cm}
  \caption{Summary of different properties for our hypergraph (left, $150$ instances) and graph benchmark set (right, $25$ instances).
           It shows for each (hyper)graph (points), the number of nodes $n$, nets/edges $m$ and pins $p$, as well as the median and maximum net size ($\medsize$ and $\maxsize{e}$),
           and node degree ($\meddeg$ and $\maxdeg{v}$). }
	\label{fig:instances}
  \vspace{-0.1cm}
\end{figure*}

\paragraph{Benchmark Set.}
%Our benchmark set is composed of hypergraphs from the well-established benchmark set of Heuer and Schlag~\cite{KAHYPAR-CA} ($488$ hypergraphs).
Our benchmark set is composed of hypergraphs from two well-established benchmark sets of Heuer and Schlag~\cite{KAHYPAR-CA} ($488$ hypergraphs)
and the Titan23 Benchmark Suite~\cite{TITAN23} (\VLSI, $22$ hypergraphs).
The hypergraphs of the former are derived from four sources encompassing three application domains:
the ISPD98 VLSI Circuit Benchmark Suite~\cite{ISPD98} (\VLSI, $18$ hypergraphs),
the DAC 2012 Routability-Driven Placement Contest~\cite{DAC} (\VLSI, $10$ hypergraphs),
the SuiteSparse Matrix Collection~\cite{SPM} (\SPM, $184$ hypergraphs), and the International SAT Competition 2014 \cite{SAT14} (\SAT, $276$ hypergraphs).
% The hypergraphs of the latter are $22$ more recent \VLSI~instances.
We include all $50$ \VLSI~instances and randomly sampled $50$ \SPM~and \SAT~instances each into our benchmark set - resulting in a total of $150$ hypergraphs.
Detailed statistics of different properties for our benchmark instances can be found in Figure~\ref{fig:instances} (left)\footnote{We
made all benchmark instances and experimental results publicly available from \anonymousurl{https://algo2.iti.kit.edu/heuer/alenex23/}{https://tinyurl.com/2ysd7ey5}}.
The largest hypergraph of the benchmark set contains $\approx 100$ million pins.

\paragraph{Experimental Setup.}
In our experiments, we map each hypergraph onto a \emph{complete rectangular $N \times M$ grid graph} $\targetarchgraph$ (similar to global routing graphs).
We construct $\targetarchgraph$ by drawing $N$ vertical and $M$ horizontal lines in a plane. The nodes are the intersections
of the lines and the edges are the connections between the intersections. The edge weights are choosen uniformly at random
between $1$ and $10$. We map each hypergraph five times using different random seeds onto
a grid graph with $k \in \{4 (2 \times 2), 8 (4 \times 2), 16 (4 \times 4), 32 (8 \times 4), 64 (8 \times 8)\}$ nodes with an allowed
imbalance of $\varepsilon = 3\%$ and a time limit of eight hours.
For each hypergraph and $k$, we use different edge weights in the corresponding grid graph.
The experiments run on a cluster of Intel Xeon Gold 6230 processors
(2 sockets with 20 cores each) running at $2.1$ GHz with 96GB RAM.
We run the different configurations of \Partitioner{Mt-KaHyPar} using ten threads.

\paragraph{Aggregating Performance Numbers.}
We call a (hyper)graph partitioned into $k$ blocks an \emph{instance}. For each instance, we aggregate running times
and the Steiner tree metric using the arithmetic mean over all seeds. To further aggregate over multiple instances, we use the geometric mean.
If all runs of an algorithm produced an imbalanced mapping (marked with \ding{55} in the plots) or ran into the time limit (marked with \ClockLogo) on
an instance, we consider the solution as \emph{infeasible}. Runs with imbalanced partitions are not excluded from aggregated running times.
For runs that exceeded the time limit, we use the time limit itself in the aggregates.
% When comparing running times,
% we say that an algorithm $\mathcal{A}$ is faster than $\mathcal{B}$ by a factor of $x$ on average if
% $x := \nicefrac{\gmean{t_\mathcal{B}}}{\gmean{t_\mathcal{A}}} > 1$ where $\gmean{t_\mathcal{A}}$ and
% $\gmean{t_\mathcal{B}}$ are geometric mean running times of $\mathcal{A}$ and $\mathcal{B}$.

\paragraph{Performance Profiles.}
\emph{Performance profiles} can be used to compare the solution quality of different algorithms~\cite{PERFORMANCE-PROFILES}.
Let $\mathcal{A}$ be the set of all algorithms, $\mathcal{I}$ the set of instances, and $q_{A}(I)$ the quality of algorithm
$A \in \mathcal{A}$ on instance $I \in \mathcal{I}$.
% ($q_{A}(I)$ is the arithmetic mean over all seeds).
For each algorithm $A$, performance profiles show the fraction of instances ($y$-axis) for which $q_A(I) \leq \tau \cdot \text{Best}(I)$, where $\tau$ is on the $x$-axis
and $\text{Best}(I)$ is the best solution produced by any of the algorithms for instance $I$.
For $\tau = 1$, the $y$-value indicates the percentage of instances for which an algorithm performs best.
Achieving higher fractions at smaller $\tau$ values is considered better.
The \ding{55}-~and \ClockLogo-tick indicates the fraction of instances for which \emph{all} runs of that algorithm produced
an imbalanced mapping or timed out.
Note that these plots relate the quality of an algorithm to the best solution and thus do not permit a full ranking of
three or more algorithms.

\begin{figure*}[!t]
  \centering
  \begin{minipage}{\textwidth}
  \ifpdfplots
    \includegraphics{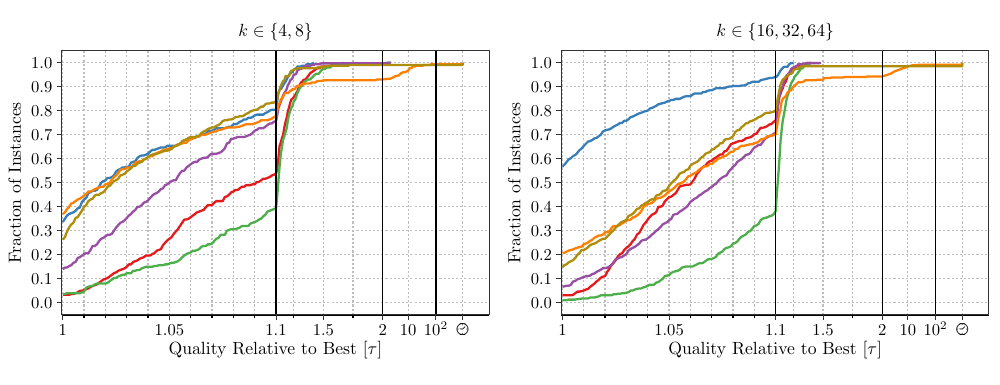}
  \else
    \tikzsetnextfilename{pdf_plots/hgp_comparison}%
    \input{tikz_plots/hgp_comparison}%
  \fi
  \end{minipage} %
  \begin{minipage}{\textwidth}
    \vspace{-0.25cm}
    \centering
  \ifpdfplots
    \includegraphics{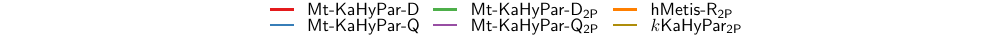}
  \else
    \tikzsetnextfilename{pdf_plots/hgp_comparison_legend}%
    \input{tikz_plots/hgp_comparison_legend}%
  \fi
  \end{minipage} %
  \vspace{-0.75cm}
  \caption{Performance profiles comparing the solution quality ($\osteiner$) of \Partitioner{Mt-KaHyPar-D/-Q} to the different two-phase approaches for
           small (left, $k \le 8$) and medium-sized target graphs (right, $k \ge 16$).}
	\label{fig:hgp_comparison}
  \vspace{-0.35cm}
\end{figure*}

\begin{figure*}[!ht]
  \centering
  \begin{minipage}{0.49\textwidth}
  \ifpdfplots
    \includegraphics{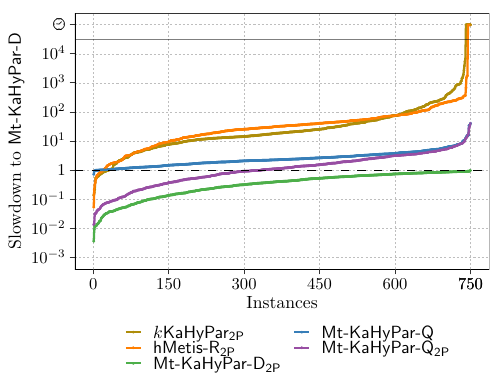}
  \else
    \tikzsetnextfilename{pdf_plots/relative_running_times_hgp}%
    \input{tikz_plots/relative_running_times_hgp}%
  \fi
  \end{minipage} %
  \begin{minipage}{0.49\textwidth}
    \centering
  \ifpdfplots
    \includegraphics{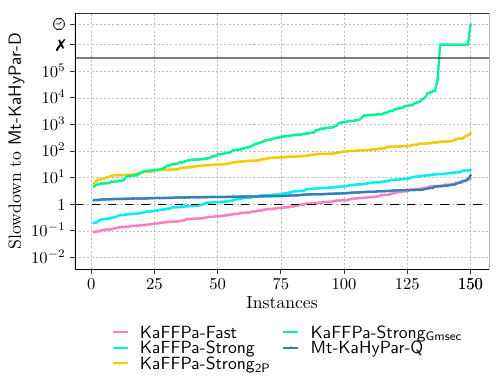}
  \else
    \tikzsetnextfilename{pdf_plots/relative_running_times_graphs}%
    \input{tikz_plots/relative_running_times_graphs}%
  \fi
  \end{minipage} %
  \vspace{-0.75cm}
  \caption{Running times of the differrent algorithms for Steiner tree metric optimization (left) and
           hierarchical process mapping (right) relative to \Partitioner{Mt-KaHyPar-D}.}
	\label{fig:relative_running_times}
\end{figure*}

\begin{figure*}[!ht]
  \centering
  \begin{minipage}{\textwidth}
  \ifpdfplots
    \includegraphics{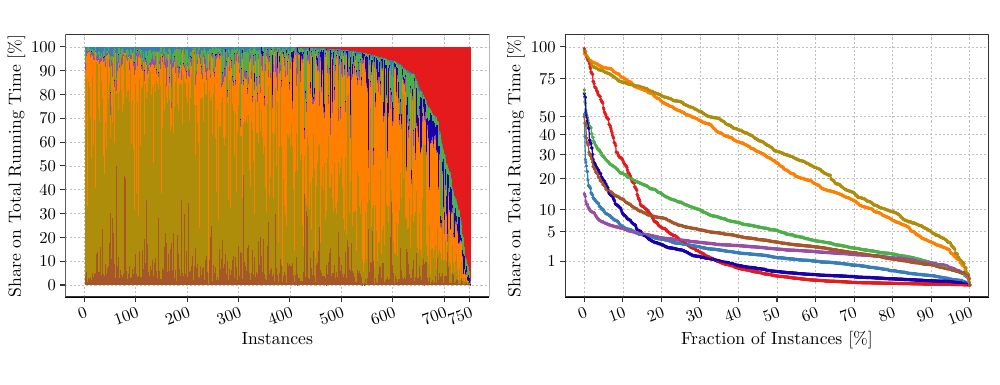}
  \else
    \tikzsetnextfilename{pdf_plots/component_breakdown}%
    \input{tikz_plots/component_breakdown}%
  \fi
  \end{minipage} %
  \begin{minipage}{\textwidth}
    \vspace{-0.25cm}
    \centering
  \ifpdfplots
    \includegraphics{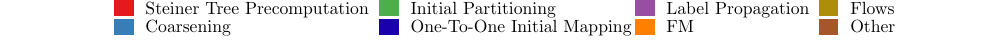}
  \else
    \tikzsetnextfilename{pdf_plots/component_breakdown_legend}%
    \input{tikz_plots/component_breakdown_legend}%
  \fi
  \end{minipage} %
  \vspace{-0.5cm}
  \caption{Running time shares of the different algorithmic components on the total execution time of \Partitioner{Mt-KaHyPar-Q}.
           The left plot shows the running time shares of each component for each instance and the right plot shows the fraction of the instances
           (x-axis) for which the share of a component on the total execution time is $\ge y\%$.}
	\label{fig:component_breakdown}
  \vspace{-0.1cm}
\end{figure*}

\begin{figure}[!t]
  \centering
  \ifpdfplots
    \includegraphics{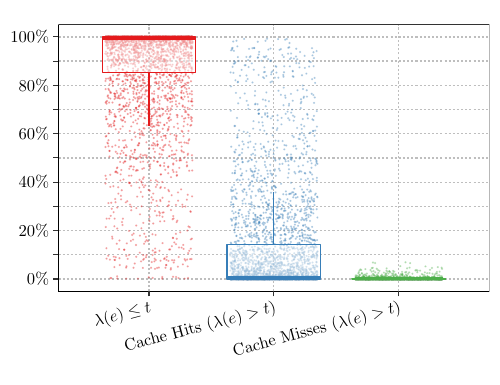}
  \else
    \tikzsetnextfilename{pdf_plots/steiner_tree_stats}%
    \input{tikz_plots/steiner_tree_stats}%
  \fi
  \vspace{-1cm}
  \caption{The percentage of Steiner tree queries in \Partitioner{Mt-KaHyPar-Q}
           that access a minimal Steiner tree ($\con(e) \le t$)
           or a $2$-approximation ($\con(e) > t$) for each instance (points). For $2$-approximations,
           we show the percentage of cache hits (accesses a cached MST) and misses (computes an MST).}
	\label{fig:steiner_tree_stats}
  \vspace{-0.25cm}
\end{figure}

\begin{filecontents*}{data/hgp_comparison.dat}
best_partitions_for_large_mt_kahypar_d = 3.11
best_partitions_for_large_mt_kahypar_d_f = 56.67
best_partitions_for_large_mt_kahypar_d_2p = 1.11
best_partitions_for_large_mt_kahypar_d_f_2p = 6.67
best_partitions_for_large_hmetis_r_2p = 20.67
best_partitions_for_large_kahypar_k_2p = 14.89
median_improvement_mt_kahypar_q_vs_d = 4.5
median_improvement_mt_kahypar_q_vs_d_2p = 10
median_improvement_mt_kahypar_q_vs_q_2p = 5.5
median_improvement_mt_kahypar_q_vs_hmetis_r_2p = 5.2
median_improvement_mt_kahypar_q_vs_kahypar_k_2p = 3.8
median_improvement_mt_kahypar_d_vs_hmetis_r_2p = 1.1
median_improvement_mt_kahypar_d_vs_kahypar_k_2p = 0.9
median_improvement_mt_kahypar_q_vs_kahypar_k_2p_k4_all = -0.7
median_improvement_mt_kahypar_q_vs_kahypar_k_2p_k4_vlsi = -1.1
median_improvement_mt_kahypar_q_vs_kahypar_k_2p_k4_spm = 0
median_improvement_mt_kahypar_q_vs_kahypar_k_2p_k4_sat = -0.8
median_improvement_mt_kahypar_q_vs_kahypar_k_2p_k8_all = 0.9
median_improvement_mt_kahypar_q_vs_kahypar_k_2p_k8_vlsi = 0.5
median_improvement_mt_kahypar_q_vs_kahypar_k_2p_k8_spm = 1.3
median_improvement_mt_kahypar_q_vs_kahypar_k_2p_k8_sat = 0.6
median_improvement_mt_kahypar_q_vs_kahypar_k_2p_k16_all = 2.9
median_improvement_mt_kahypar_q_vs_kahypar_k_2p_k16_vlsi = 4.2
median_improvement_mt_kahypar_q_vs_kahypar_k_2p_k16_spm = 1.7
median_improvement_mt_kahypar_q_vs_kahypar_k_2p_k16_sat = 2.1
median_improvement_mt_kahypar_q_vs_kahypar_k_2p_k32_all = 4.7
median_improvement_mt_kahypar_q_vs_kahypar_k_2p_k32_vlsi = 5.5
median_improvement_mt_kahypar_q_vs_kahypar_k_2p_k32_spm = 3.7
median_improvement_mt_kahypar_q_vs_kahypar_k_2p_k32_sat = 3.4
median_improvement_mt_kahypar_q_vs_kahypar_k_2p_k64_all = 4.4
median_improvement_mt_kahypar_q_vs_kahypar_k_2p_k64_vlsi = 7
median_improvement_mt_kahypar_q_vs_kahypar_k_2p_k64_spm = 2.7
median_improvement_mt_kahypar_q_vs_kahypar_k_2p_k64_sat = 3
gmean_time_mt_kahypar_d = 2.6
gmean_time_mt_kahypar_q = 6.5
gmean_time_mt_kahypar_d_2p = 0.8
gmean_time_mt_kahypar_q_2p = 2.7
gmean_time_kahypar_k_2p = 59.2
gmean_time_hmetis_r_2p = 71.1
slowdown_factor_mt_kahypar_d_vs_d_2p = 3.2
slowdown_factor_mt_kahypar_q_vs_q_2p = 2.3
\end{filecontents*}
\DTLloaddb[noheader, keys={key,value}]{hgp}{data/hgp_comparison.dat}

\paragraph{Steiner Tree Metric Optimization.}
Figure~\ref{fig:hgp_comparison} compares the mapping quality ($\osteiner$) produced by the
different algorithms. As we can see, two-phase approaches produce similar mappings when compared to our new approach for smaller
target graphs ($k \in \{4,8\}$). We therefore focus on the results for larger target graphs ($k \in \{16,32,64\}$).
For these graphs, \Partitioner{Mt-KaHyPar-Q} finds the best mappings on most of the instances ($\placeholder{hgp}{best_partitions_for_large_mt_kahypar_d_f}$\%)
followed by \Partitioner{hMetis-R$_{2P}$} ($\placeholder{hgp}{best_partitions_for_large_hmetis_r_2p}$\%) and
\Partitioner{$k$KaHyPar$_{2P}$} ($\placeholder{hgp}{best_partitions_for_large_kahypar_k_2p}$\%).
The median improvement of \Partitioner{Mt-KaHyPar-Q} over \Partitioner{Mt-KaHyPar-D$_{2P}$}, \Partitioner{Mt-KaHyPar-Q$_{2P}$},
\Partitioner{hMetis-R$_{2P}$}, \Partitioner{Mt-KaHyPar-D}, and \Partitioner{$k$KaHyPar$_{2P}$} is
$\placeholder{hgp}{median_improvement_mt_kahypar_q_vs_d_2p}$\%, $\placeholder{hgp}{median_improvement_mt_kahypar_q_vs_q_2p}$\%,
$\placeholder{hgp}{median_improvement_mt_kahypar_q_vs_hmetis_r_2p}$\%, $\placeholder{hgp}{median_improvement_mt_kahypar_q_vs_d}$\%, and
$\placeholder{hgp}{median_improvement_mt_kahypar_q_vs_kahypar_k_2p}$\%.
Our default configuration computes better mappings than \Partitioner{hMetis-R$_{2P}$} by $\placeholder{hgp}{median_improvement_mt_kahypar_d_vs_hmetis_r_2p}$\% in the median,
while \Partitioner{$k$KaHyPar$_{2P}$} produces slightly better results than \Partitioner{Mt-KaHyPar-D}
($\placeholder{hgp}{median_improvement_mt_kahypar_d_vs_kahypar_k_2p}$\% in the median).
Table~\ref{tbl:improvements} summarizes median improvements of \Partitioner{Mt-KaHyPar-Q} over \Partitioner{$k$KaHyPar$_{2P}$}
for the different instance types and target graphs ($k$). It can be seen that larger target graphs lead to
larger improvements, while improvements are most pronounced on \VLSI~instances.

\begin{table}[!t]
  \centering
  \caption{Median improvements in percentage of \Partitioner{Mt-KaHyPar-Q} over \Partitioner{$k$KaHyPar$_{2P}$} for the different
           instance types and $k$.}\label{tbl:improvements}
  \vspace{0.1cm}
  \begin{tabular}{rccccc}
  %  &  \multicolumn{5}{c}{Median Improvement of} \\
  %  &  \multicolumn{5}{c}{\Partitioner{Mt-KaHyPar-Q} over \Partitioner{$k$KaHyPar$_{2P}$} $[\%]$} \\
    & $k = 4$ & $k = 8$ & $k = 16$ & $k = 32$ & $k = 64$ \\
  \midrule
  \VLSI & $\placeholder{hgp}{median_improvement_mt_kahypar_q_vs_kahypar_k_2p_k4_vlsi}$ & $\placeholder{hgp}{median_improvement_mt_kahypar_q_vs_kahypar_k_2p_k8_vlsi}$ &
          $\placeholder{hgp}{median_improvement_mt_kahypar_q_vs_kahypar_k_2p_k16_vlsi}$ & $\placeholder{hgp}{median_improvement_mt_kahypar_q_vs_kahypar_k_2p_k32_vlsi}$ &
          $\placeholder{hgp}{median_improvement_mt_kahypar_q_vs_kahypar_k_2p_k64_vlsi}$ \\
  \SPM & $\placeholder{hgp}{median_improvement_mt_kahypar_q_vs_kahypar_k_2p_k4_spm}$ & $\placeholder{hgp}{median_improvement_mt_kahypar_q_vs_kahypar_k_2p_k8_spm}$ &
         $\placeholder{hgp}{median_improvement_mt_kahypar_q_vs_kahypar_k_2p_k16_spm}$ & $\placeholder{hgp}{median_improvement_mt_kahypar_q_vs_kahypar_k_2p_k32_spm}$ &
         $\placeholder{hgp}{median_improvement_mt_kahypar_q_vs_kahypar_k_2p_k64_spm}$ \\
  \SAT & $\placeholder{hgp}{median_improvement_mt_kahypar_q_vs_kahypar_k_2p_k4_sat}$ & $\placeholder{hgp}{median_improvement_mt_kahypar_q_vs_kahypar_k_2p_k8_sat}$ &
         $\placeholder{hgp}{median_improvement_mt_kahypar_q_vs_kahypar_k_2p_k16_sat}$ & $\placeholder{hgp}{median_improvement_mt_kahypar_q_vs_kahypar_k_2p_k32_sat}$ &
         $\placeholder{hgp}{median_improvement_mt_kahypar_q_vs_kahypar_k_2p_k64_sat}$ \\
  \midrule
  \ALL & $\placeholder{hgp}{median_improvement_mt_kahypar_q_vs_kahypar_k_2p_k4_all}$ & $\placeholder{hgp}{median_improvement_mt_kahypar_q_vs_kahypar_k_2p_k8_all}$ &
         $\placeholder{hgp}{median_improvement_mt_kahypar_q_vs_kahypar_k_2p_k16_all}$ & $\placeholder{hgp}{median_improvement_mt_kahypar_q_vs_kahypar_k_2p_k32_all}$ &
         $\placeholder{hgp}{median_improvement_mt_kahypar_q_vs_kahypar_k_2p_k64_all}$ \\
  \end{tabular}
  \vspace{-0.6cm}
\end{table}

The different algorithms can be ranked from fastest to slowest as follows:
\Partitioner{Mt-KaHyPar-D$_{2P}$} (\gmeantime~$\placeholder{hgp}{gmean_time_mt_kahypar_d_2p}$s),
\Partitioner{Mt-KaHyPar-D} ($\placeholder{hgp}{gmean_time_mt_kahypar_d}$s),
\Partitioner{Mt-KaHyPar-Q$_{2P}$} ($\placeholder{hgp}{gmean_time_mt_kahypar_q_2p}$s),
\Partitioner{Mt-KaHyPar-Q} ($\placeholder{hgp}{gmean_time_mt_kahypar_q}$s),
\Partitioner{$k$KaHyPar$_{2P}$} ($\placeholder{hgp}{gmean_time_kahypar_k_2p}$s), and
\Partitioner{hMetis-R$_{2P}$} ($\placeholder{hgp}{gmean_time_hmetis_r_2p}$s).
As we can see, optimizing the Steiner tree metric introduces a slowdown by a factor of
$\placeholder{hgp}{slowdown_factor_mt_kahypar_d_vs_d_2p}$ (\Partitioner{Mt-KaHyPar-D}) and
$\placeholder{hgp}{slowdown_factor_mt_kahypar_q_vs_q_2p}$ (\Partitioner{Mt-KaHyPar-Q}) on average
when compared to their corresponding two-phase configurations that optimizes the connectivity metirc.
However, our best configuration \Partitioner{Mt-KaHyPar-Q} is almost an order of magnitude
faster than \Partitioner{$k$KaHyPar$_{2P}$} when using ten threads.
In Figure~\ref{fig:relative_running_times} (left),
we show the running times of the different algorithms relative to \Partitioner{Mt-KaHyPar-D}
for each instance.

Figure~\ref{fig:component_breakdown} shows running time shares
of the different algorithmic components on the total execution time of \Partitioner{Mt-KaHyPar-Q}.
The most time-consuming components are FM ($26.3\%$ in the median) and flow-based refinement ($31.5\%$).
The Steiner tree precomputation step takes the most time on $8.5\%$ of the instances (primarily on smaller
hypergraphs mapped onto larger target graphs).

\begin{filecontents*}{data/steiner_tree_stats.dat}
median_opt_steiner_tree = 99.1
median_cache_hits = 0.8
median_cache_misses = 0.00017
median_cache_hits_2 = 99.8
median_cache_misses_2 = 0.2042
max_cache_misses = 6.9
opt_steiner_tree_metric = 45.1
one_percent_steiner_tree_metric = 93.4
\end{filecontents*}
\DTLloaddb[noheader, keys={key,value}]{stats}{data/steiner_tree_stats.dat}

\paragraph{Accuracy of $\osteiner$.}
Figure~\ref{fig:steiner_tree_stats} shows that most of the Steiner tree queries access a
minimal Steiner tree ($\con(e) \le t$, $\placeholder{stats}{median_opt_steiner_tree}\%$ in the median).
For queries with $\con(e) > t$ ($2$-approximation, $\placeholder{stats}{median_cache_hits}\%$),
we return a cached MST for $\placeholder{stats}{median_cache_hits_2}\%$ (median) of the queries.
This confirms our choice for $t$ and demonstrates the effectiveness of the caching mechanism.
However, larger cache miss rates can be observed on instances with large hyperedges (primarily \SPM~instances)
mapped onto larger target graphs (largest cache miss rate is $\placeholder{stats}{max_cache_misses}\%$).

For the final mappings generated by \Partitioner{Mt-KaHyPar-Q}, we are able to compute the optimal value of $\osteiner$ for
$\placeholder{stats}{opt_steiner_tree_metric}\%$ of instances. Moreover, we achieve a approximation ratio within $1\%$ of the optimal
value of $\osteiner$ for $\placeholder{stats}{one_percent_steiner_tree_metric}\%$ of the instances, meaning that we use a $2$-approximation for at most
$1\%$ of the nets.

\begin{filecontents*}{data/graph_comparison.dat}
median_improvement_mt_kahypar_q_vs_strong = 2.8
median_improvement_mt_kahypar_q_vs_fast = 6.1
median_improvement_mt_kahypar_q_vs_strong_gmsec = 8
median_improvement_mt_kahypar_q_vs_strong_2p = 6.7
median_improvement_mt_kahypar_q_vs_d = 8.5
gmean_time_mt_kahypar_d = 3.7
gmean_time_mt_kahypar_q = 9.1
gmean_time_kaffpa_fast = 2.9
gmean_time_kaffpa_strong = 9
gmean_time_kaffpa_strong_2p = 204.8
gmean_time_kaffpa_strong_gmsec = 774.4
\end{filecontents*}
\DTLloaddb[noheader, keys={key,value}]{graph}{data/graph_comparison.dat}

\paragraph{Hierarchical Process Mapping.}
Mapping graphs onto hierarchical processor architectures is a field that has seen
a lot of progress recently~\cite{KAFFPA-IMAP,PARHIP-IMAP,Kirchbach0T20}. These architectures
are described by a sequence $a_1:a_2:\ldots:a_l$ which can be interpreted as each processor
having $a_1$ cores, each node $a_2$ processors, each rack $a_3$ nodes, and so forth. The
communication costs between the cores can be described by a sequence $d_1:d_2:\ldots:d_l$, meaning
that two cores in the same processor communicate with cost $d_1$, two cores in the same node but different
processors communicate with cost $d_2$, and so forth. Since $\osteiner$ reverts to the objective
function for process mapping on graphs, we compare \Partitioner{Mt-KaHyPar-D/-Q}
to the best performing algorithms of Ref.~\cite{KAFFPA-IMAP}. This includes \Partitioner{KaFFPa-Fast/-Strong}
(multilevel algorithm that directly optimizes $\osteiner$, similar to \Partitioner{Mt-KaHyPar-D/-Q}), \Partitioner{KaFFPa-Strong$_{2P}$}
(best two-phase approach), and \Partitioner{KaFFPa-Strong$_{\text{Gmsec}}$} (highest-quality algorithm).
\Partitioner{KaFFPa-Strong$_{\text{Gmsec}}$} optimizes the cut-net metric, but uses global multisectioning (Gmsec) to partition the graph along the hierarchy,
meaning it first partitions the graph into $a_1$ blocks, then each block further into $a_2$ blocks, and so forth.
For the experiments, we use the instances from their paper ($25$ graphs) and map them onto processor architectures described by the
sequence $4:8:r$ with $r \in \{1,2,\ldots,6\}$ and communication costs $1:10:100$ ($5$ repetitions, $\varepsilon = 3\%$, $8$ hours
time limit). Figure~\ref{fig:instances} (right) summarizes different properties of the benchmark instances.

\begin{figure}[!t]
  \centering
  \ifpdfplots
    \includegraphics{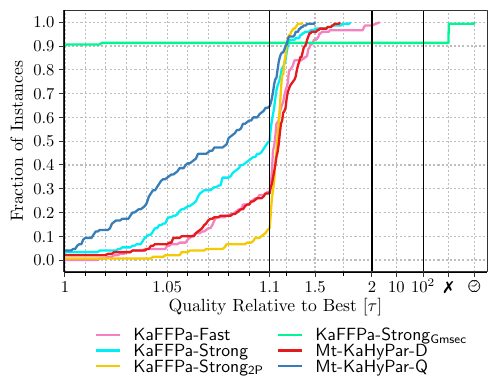}
  \else
    \tikzsetnextfilename{pdf_plots/hierarchical_process_mapping}%
    \input{tikz_plots/hierarchical_process_mapping}%
  \fi
  \vspace{-0.75cm}
  \caption{Performance profiles comparing the solution quality ($\osteiner$) of \Partitioner{Mt-KaHyPar-D/-Q} to the different algorithms
           for hierarchical process mapping.}
	\label{fig:hierarchical_process_mapping}
  \vspace{-0.25cm}
\end{figure}

As it can be seen in Figure~\ref{fig:hierarchical_process_mapping}, \Partitioner{KaFFPa-Strong$_{\text{Gmsec}}$}
is the best performing algorithm for hierarchical process mapping. It computes better mappings than \Partitioner{Mt-KaHyPar-Q}
by $\placeholder{graph}{median_improvement_mt_kahypar_q_vs_strong_gmsec}\%$ in the median, but it is almost two orders of magnitude slower
(\gmeantime~$\placeholder{graph}{gmean_time_mt_kahypar_q}$s vs $\placeholder{graph}{gmean_time_kaffpa_strong_gmsec}$s).
This is on par with the results of Ref.~\cite{KAFFPA-IMAP} where the authors also note that the improvements of \Partitioner{KaFFPa-Strong$_{\text{Gmsec}}$}
over other algorithms becomes less pronounced when mapping graphs onto larger processor architectures.
% Moreover, the global multisectioning approach can be not generalized to abitrary processor architectures.

The median improvement of \Partitioner{Mt-KaHyPar-Q} over \Partitioner{Mt-KaHyPar-D},
\Partitioner{KaFFPa-Strong$_{2P}$}, \Partitioner{KaFFPa-Fast}, and \Partitioner{KaFFPa-Strong} is
$\placeholder{graph}{median_improvement_mt_kahypar_q_vs_d}\%$, $\placeholder{graph}{median_improvement_mt_kahypar_q_vs_strong_2p}\%$,
$\placeholder{graph}{median_improvement_mt_kahypar_q_vs_fast}\%$, and $\placeholder{graph}{median_improvement_mt_kahypar_q_vs_strong}\%$.
% These improvements can be attributed to flow-based refinement, as we introduce the first flow network model for optimizing $\osteiner$ directly.
\Partitioner{Mt-KaHyPar-Q} ($\placeholder{graph}{gmean_time_mt_kahypar_q}$s, $10$ threads) is slower than
\Partitioner{KaFFPa-Fast} ($\placeholder{graph}{gmean_time_kaffpa_fast}$s), and \Partitioner{Mt-KaHyPar-D} ($\placeholder{graph}{gmean_time_mt_kahypar_d}$s, $10$ threads),
but its running time is comparable to \Partitioner{KaFFPa-Strong} ($\placeholder{graph}{gmean_time_kaffpa_strong}$s) and
it is more than an order of magnitude faster than \Partitioner{KaFFPa-Strong$_{2P}$} ($\placeholder{graph}{gmean_time_kaffpa_strong_2p}$s).
% The slower running times are expected since some parts of the framework are not optimized for graphs.
More details on the running times of the different algorithms can be found in Figure~\ref{fig:relative_running_times} (right).

\section{Conclusion and Future Work}\label{sec:conclusion}
In this paper, we introduced a HGP formulation based on the Steiner tree metric to
accurately model wire-lengths in VLSI design and communication costs in distributed systems.
Our main algorithmic contribution is a direct $k$-way multilevel algorithm for optimizing the Steiner tree metric.
We implemented an efficient gain table data structure for constant-time lookup of gain values.
Furthermore, we devised a novel flow network model to optimize the metric via max-flow min-cut computations.
Our experimental results demonstrates that our new algorithm has the potential to significantly reduce wire-lengths in VLSI design.

Future research includes a more efficient and memory-friendly algorithm for precomputing all Steiner trees,
enabling the mapping of hypergraphs onto larger target graphs.
Additionally, we want to develop a placement and routing algorithm for VLSI design leveraging our new algorithm.

\bibliography{steiner_trees}

\end{document}